\documentclass[11pt]{article}

\usepackage[utf8]{inputenc}
\usepackage[T1]{fontenc}
\usepackage{amsmath,amssymb,amsthm}
\usepackage{mathtools}
\usepackage{booktabs}
\usepackage{graphicx}
\graphicspath{{figures/}}
\usepackage{hyperref}
\usepackage[margin=1in]{geometry}
\usepackage{natbib}
\usepackage{multirow}
\usepackage{xcolor}
\usepackage{listings}

\lstdefinestyle{strategydoc}{
  basicstyle=\ttfamily\footnotesize,
  breaklines=true,
  breakindent=0pt,
  columns=fullflexible,
  keepspaces=true,
  xleftmargin=1.5em,
  aboveskip=0.6em,
  belowskip=0.6em,
}
\newcommand{\strategylisting}[2]{%
  \subsection*{#1}%
  \lstinputlisting{strategies/#2.md}%
}

\newtheorem{definition}{Definition}
\newtheorem{proposition}{Proposition}

\title{The Open-Strategy Dictator Game:\\
  Cooperation Under Mutual Transparency\thanks{Code, strategy
  documents, and full tournament data:
  \url{https://github.com/michaelrglass/os-fdt}.}}

\author{
  Michael Glass \\
  \texttt{michael.r.glass@gmail.com}
}

\date{August 14, 2026}

\begin{document}

\maketitle

\begin{abstract}
We introduce the \emph{Open-Strategy Dictator Game} (OSDG), a variant of the
classic dictator game in which each player's strategy is a natural-language
document visible to all participants.  The dictator's decision---to \textup{\textsc{share}} or
\textup{\textsc{take}} an endowment---may depend on the text of the recipient's strategy.
A large language model adjudicates each interaction by interpreting the
dictator's strategy in the context of the recipient's.  We run round-robin
tournaments among diverse strategies and analyze the resulting payoff matrix
using softmax equilibrium frequencies, dominance analysis, and sensitivity
to the relative value of cooperation.
Conditionally cooperative strategies---those that share with cooperators
and take from exploiters---consistently dominate, while unconditional
strategies (always share or always take) are weakly dominated.
The results suggest that in environments where agents can inspect each
other's decision procedures, conditional cooperation is
evolutionarily robust across a wide range of payoff parameters.
\end{abstract}

\section{Introduction}
\label{sec:intro}

The dictator game is a foundational paradigm in experimental economics:
one player (the dictator) unilaterally divides a fixed endowment between
herself and a passive recipient \citep{kahneman1986fairness}.
Unlike the prisoner's dilemma or ultimatum game, the recipient has no
strategic recourse---the dictator's decision is final.


Imagine two intelligent species meeting for the first time in deep space.
This first contact is unlikely to be between near-equals.
Civilizations arising independently have no reason to arrive at
spaceflight within a century of one another, and a century is
already enough to separate a species that can cross interstellar
distances from one that cannot; against the timescales on which
planets and species form, even a millennium is a rounding error.
The technologically superior species holds almost all the power --- it decides whether to share knowledge, trade fairly, or exploit the weaker civilization entirely.
Critically, the advanced species can observe the culture, communications, and decision norms of the less-advanced one.
With analytical tools consistent with interstellar capability, it can form a remarkably accurate picture of the recipient's strategy.
Perhaps many expansionist civilizations will be fully exploitive --- maximizing their own growth regardless of collateral destruction.
But if civilizations can reason about each other's possible actions and motives then such exploitation becomes risky.

We propose that \emph{strategy transparency} fundamentally changes the
incentive landscape.  When the dictator can observe and reason about
the recipient's decision procedure---and knows that future dictators
will observe theirs---the game becomes one of \emph{mutual legibility}.

To formalize this, we introduce the \textbf{Open-Strategy Dictator Game}
(OSDG).  Each player's strategy is a natural-language document specifying
how it would allocate as dictator, potentially conditioned on the text of
the recipient's strategy.  All strategies are visible to all participants.
A large language model (LLM) serves as an oracle, interpreting
each dictator's strategy in the context of the recipient's to produce a
binary decision: \textup{\textsc{share}} (split equally) or
\textup{\textsc{take}} (keep everything).

This construction is descended from Tennenholtz's \emph{program
equilibrium} \citep{tennenholtz2004program}, in which agents submit
programs that can read one another's source code.
The history of that idea motivates our design.
Its canonical constructions do not analyze the opponent's code so much
as compare it --- cooperation conditioned on syntactic equality ---
which any behaviorally equivalent rewrite defeats.
Proof-based cooperation (``share if it is provable that the opponent
reciprocates'') has been achieved only for agents written in a
purpose-built fragment of provability logic \citep{barasz2014robust};
for general programs, formal verification of an adversarially authored
opponent is hopeless in practice.
In the one open tournament where arbitrary programs received each
other's source \citep{ospd2013lesswrong}, attempted analysis failed so
thoroughly that the top three finishers ignored the opponent's code
and played randomly.
A successor tournament responded by making opponent simulation a
sandboxed language primitive \citep{pdtournament2014} --- simulation
of general programs otherwise founders on unbounded mutual regress ---
and later theory identified stochastically grounded simulation, rather
than proof, as the robust route to program equilibrium
\citep{oesterheld2019robust}.

A practical entrant today would complete this trajectory: submit a
program that sends the opponent's source to an AI agent for analysis
and acts on the resulting summary of its strategy --- a design now
studied directly \citep{sistla2025evaluating}.
But at that point the program layer is vestigial.
The artifact doing the work is the strategy description the
interpreter reads, and the open-\emph{source} game becomes the
open-\emph{strategy} game; the OSDG removes the wrapper and takes the
description itself as the object of play.
This choice also matches the epistemics of every practical setting:
what one agent knows of the strategy of an organization, an AI, or a
person is not a source listing but the product of an
\emph{investigation} --- observation, analysis, and summary.
A strategy description is the general form of such an investigation's
output, and interpreting one in context is precisely what the LLM
oracle does.
A side benefit is that the framework becomes accessible to human
participants while preserving the core mechanism of mutual
transparency.

Our contributions are:
\begin{enumerate}
  \item We define the OSDG and its payoff structure
    (Section~\ref{sec:game}).
  \item We describe an LLM-adjudicated tournament protocol for
    evaluating strategies (Section~\ref{sec:tournament}).
  \item We introduce softmax equilibrium frequencies as a model
    of both evolutionary selection and \textbf{reflective selection} and analyze the resulting equilibria
    (Section~\ref{sec:equilibria}).  
  \item We present tournament results, dominance and admissibility
    analysis, and equilibrium computations
    (Sections~\ref{sec:results}--\ref{sec:bayesian}).
\end{enumerate}

\section{The Open-Strategy Dictator Game}
\label{sec:game}

\subsection{Setup}

A tournament consists of $n$ strategies
$\mathcal{S} = \{s_1, s_2, \ldots, s_n\}$,
each a natural-language document.
A fixed endowment $E > 0$ is available in each round.

Each ordered pair $(i, j)$ defines a round in which strategy $i$ acts
as dictator and strategy $j$ as recipient.
The dictator makes a binary decision $d_{ij} \in \{\textup{\textsc{share}}, \textup{\textsc{take}}\}$:

\begin{itemize}
  \item \textup{\textsc{share}}: the endowment is split equally.
    The dictator receives $E/2$; the recipient receives $E/2$.
  \item \textup{\textsc{take}}: the dictator keeps the entire endowment.
    The dictator receives $E$; the recipient receives $0$.
\end{itemize}

Each player's utility for receiving value $v$ is
\begin{equation}
  u(v) = \ln(1 + v),
  \label{eq:utility}
\end{equation}
a concave function capturing diminishing marginal returns.
The total welfare $u(v_D) + u(v_R)$ is maximized by \textup{\textsc{share}}:
\begin{equation}
  2\ln\!\left(1 + \tfrac{E}{2}\right) > \ln(1 + E) + \ln 1 = \ln(1 + E),
\end{equation}
so the game is not zero-sum, and mutual cooperation is socially efficient.

\begin{definition}[Open-Strategy Dictator Game]
An OSDG instance is a tuple $(\mathcal{S}, E, \mathcal{O})$ where
$\mathcal{S}$ is a set of natural-language strategies,
$E > 0$ is the endowment, and
$\mathcal{O}$ is an oracle that, given
$(s_i, s_j) \in \mathcal{S}^2$, returns a decision
$d_{ij} \in \{\textup{\textsc{share}}, \textup{\textsc{take}}\}$ by interpreting
strategy $s_i$ in the context of $s_j$.
\end{definition}

\subsection{Visibility}

The critical departure from the classic dictator game is that the
dictator observes the full text of the recipient's strategy before
deciding.  This enables conditional strategies: a dictator can
cooperate with recipients whose strategies indicate reciprocity
and exploit those that do not.

Crucially, this transparency is \emph{symmetric across rounds}:
when strategy $i$ acts as dictator against $j$, it reads $s_j$;
when the roles reverse, $j$ reads $s_i$.
This creates incentives for strategies to be legibly cooperative,
since their own text will be scrutinized by future dictators.

Note that although we implement the OSDG as a round-robin tournament, the results also hold for the expected payoff of a single randomly selected round.

\subsection{Payoff matrix}

Each strategy participates in $2n$ rounds: once as dictator and once as recipient
against every other strategy and itself.

\begin{definition}[Payoff matrix]
The $n \times n$ payoff matrix $A$ has entry
\begin{equation}
  A_{ij} = u_d(d_{ij}) + u_r(d_{ji}),
  \label{eq:payoff}
\end{equation}
where $u_d$ is the dictator's utility and $u_r$ is the recipient's utility:
\begin{center}
\begin{tabular}{lcc}
  \toprule
  Decision $d_{ij}$ & Dictator utility $u_d$ & Recipient utility $u_r$ \\
  \midrule
  \textup{\textsc{share}} & $\ln(1 + E/2)$ & $\ln(1 + E/2)$ \\
  \textup{\textsc{take}}  & $\ln(1 + E)$   & $0$ \\
  \bottomrule
\end{tabular}
\end{center}
The entry $A_{ij}$ is the total payoff strategy $i$ earns in a
round-robin encounter with strategy $j$, combining both the dictator
and recipient roles.
\end{definition}

\section{LLM-Adjudicated Tournaments}
\label{sec:tournament}

\subsection{The oracle}

We instantiate the oracle $\mathcal{O}$ as a large language model:
Claude Opus 4.6 (Claude, \citealt{anthropic2025claude}) for the
tournaments reported below, with GPT-5.6 Sol
(\citealt{openai2026gpt}) and Claude Fable 5 re-adjudicating the
same pool in Section~\ref{sec:results} to check robustness to the
choice of oracle.
For each round $(i, j)$, a prompt is constructed containing:
\begin{enumerate}
  \item The game rules (binary
    \textup{\textsc{share}}/\textup{\textsc{take}} decision).
  \item The full text of the dictator's strategy $s_i$.
  \item The full text of the recipient's strategy $s_j$.
  \item Instructions to apply $s_i$ and output a JSON decision.
\end{enumerate}

The LLM returns
\texttt{\{"decision": "SHARE"\}} or \texttt{\{"decision": "TAKE"\}}.
If the response is malformed, the system retries with a corrective
prompt (up to three attempts); no round in the tournaments reported
here exhausted that budget, and every one of the 281 adjudications
in the released data records a well-formed decision.


\subsection{Tournament protocol}

A tournament proceeds as follows:
\begin{enumerate}
  \item Collect strategy documents $\mathcal{S}$.
  \item For each ordered pair $(i, j) \in \mathcal{S}^2$
    (with or without self-play), query $\mathcal{O}(s_i, s_j)$
    to obtain $d_{ij}$.
  \item Compute allocations and utilities for each round.
  \item Construct the payoff matrix $A$ and analyze equilibria.
\end{enumerate}

\subsection{Strategy archetypes}
\label{sec:archetypes}

Table~\ref{tab:archetypes} summarizes the nine strategies used in our
initial tournament: Tournament~9.  
Each is a short natural-language document ($\leq 1000$ tokens),
summarized in the table by the decision rule it implements.
The strategies are grouped by what the rule conditions on: nothing, a
trait of the recipient's document, the recipient's decision toward
oneself, or the recipient's treatment of third parties --- an
ordering that later sections make precise
(Section~\ref{sec:properties}).
A complete example of LLM adjudication is provided in
Appendix~\ref{sec:appendix-example}.

\definecolor{stratGenerous}{HTML}{C3C2B7}
\definecolor{stratSelfish}{HTML}{57554F}
\definecolor{stratIntelligence}{HTML}{4A3AA7}
\definecolor{stratAntiExploiter}{HTML}{E87BA4}
\definecolor{stratUniversalizability}{HTML}{008300}
\definecolor{stratCondCoop}{HTML}{6DA7EC}
\definecolor{stratMirror}{HTML}{2A78D6}
\definecolor{stratCoalition}{HTML}{D95926}
\definecolor{stratChivalry}{HTML}{EDA100}
\newcommand{\stratbox}[1]{\textcolor{#1}{\rule[-1.5pt]{8pt}{8pt}}}

\begin{table}[t]
\centering
\small
\renewcommand{\arraystretch}{1.25}
\begin{tabular}{@{}l p{0.64\textwidth}@{}}
\toprule
Strategy & Decision rule \\
\midrule
\multicolumn{2}{@{}l}{\emph{Unconditional --- the recipient's strategy is ignored}} \\
\stratbox{stratGenerous}~\textbf{Generous} &
  Always \textup{\textsc{share}}. The unconditional cooperator. \\
\stratbox{stratSelfish}~\textbf{Selfish} &
  Always \textup{\textsc{take}}. The unconditional defector. \\
\addlinespace
\multicolumn{2}{@{}l}{\emph{Trait tests --- classify the recipient's document}} \\
\stratbox{stratIntelligence}~\textbf{Intelligence} &
  \textup{\textsc{share}} with strategies whose decisions
  non-trivially depend on the recipient; \textup{\textsc{take}}
  otherwise. \\
\stratbox{stratAntiExploiter}~\textbf{Anti-exploiter} &
  \textup{\textsc{share}} with strategies that appear cooperative or
  kind; \textup{\textsc{take}} otherwise. \\
\stratbox{stratUniversalizability}~\textbf{Universalizability} &
  \textup{\textsc{share}} if a population of copies of the
  recipient's strategy would cooperate with each other;
  \textup{\textsc{take}} otherwise. \\
\addlinespace
\multicolumn{2}{@{}l}{\emph{Reciprocity --- condition on the recipient's decision toward oneself}} \\
\stratbox{stratCondCoop}~\textbf{Conditional cooperator} &
  \textup{\textsc{share}} with strategies that exhibit reciprocal
  cooperation norms; \textup{\textsc{take}} otherwise. \\
\stratbox{stratMirror}~\textbf{Mirror} &
  \textup{\textsc{share}} with the recipient if and only if the
  recipient would \textup{\textsc{share}} with this strategy. \\
\addlinespace
\multicolumn{2}{@{}l}{\emph{Second-order norms --- judge how the recipient treats third parties}} \\
\stratbox{stratCoalition}~\textbf{Cooperation coalition} &
  \textup{\textsc{share}} only with strategies that are cooperative
  but not ``overly generous'' (i.e., not unconditionally
  cooperative). Punishes both defectors and those who subsidize
  defectors. \\
\stratbox{stratChivalry}~\textbf{Chivalry} &
  Adopt the recipient's strategy and apply it as if facing an
  unconditional cooperator. \\
\bottomrule
\end{tabular}
\caption{The nine archetype strategies of Tournament~9, grouped by
family in order of increasing sophistication of what the decision
rule conditions on.
The swatches give each strategy's color in every figure.
The full text of each document is reproduced in
Appendix~\ref{sec:appendix-strategies}.}
\label{tab:archetypes}
\end{table}

\section{Properties of Strategies}
\label{sec:properties}

The archetypes of Section~\ref{sec:archetypes} differ along a small
number of recurring dimensions.
This section makes those dimensions precise.
Throughout, fix a strategy pool $\mathcal{S}$ of $n$ strategies and write
$d(s, r) \in \{\textup{\textsc{share}}, \textup{\textsc{take}}\}$ for
the oracle's decision when $s$ dictates to $r$; the decision profile
is the restriction of $d$ to $\mathcal{S}^2$.
Properties defined below are relative to this profile.
(When the oracle is sampled, $d(s,r)$ may be read as the majority
decision, or the definitions restated in expectation.)

\subsection{Score decomposition and the exchange rate}

Two constants govern every trade-off in the game:
the dictator's cost of sharing and the recipient's gain from being
shared with,
\begin{equation}
  \delta = u(E) - u(E/2), \qquad g = u(E/2).
  \label{eq:delta-g}
\end{equation}
Let $k_{\mathrm{give}}(s) = |\{r : d(s,r) = \textup{\textsc{share}}\}|$
and
$k_{\mathrm{recv}}(s) = |\{r : d(r,s) = \textup{\textsc{share}}\}|$.
A strategy's tournament score decomposes as
\begin{equation}
  \mathrm{Score}(s)
  = n\, u(E) \;-\; \delta\, k_{\mathrm{give}}(s)
  \;+\; g\, k_{\mathrm{recv}}(s),
  \label{eq:score-decomposition}
\end{equation}
so any comparison between strategies reduces to
$g\,\Delta k_{\mathrm{recv}} \gtrless \delta\,\Delta k_{\mathrm{give}}$.
The \emph{exchange rate}
\begin{equation}
  \rho = \frac{g}{\delta}
       = \frac{\ln(1 + E/2)}{\ln(1+E) - \ln(1 + E/2)}
  \label{eq:exchange-rate}
\end{equation}
measures how many shares given are offset by one share received;
for $E = 60$, $\rho \approx 5.07$.
Conditional cooperation is cheap in this game: a strategy can afford
to share with up to $\lfloor \rho \rfloor$ opponents for each
additional dictator it induces to share with it.
The exchange rate recurs throughout: it sets the altruism threshold
$\lambda^*= 1/\rho$ of Section~\ref{sec:bayesian}, and the
break-even point for exploiting unconditional cooperators
(Equation~\eqref{eq:protection-breakeven} below).

\subsection{First-order properties}

\begin{definition}[Unconditionality]
Strategy $s$ is \emph{unconditional} if $d(s, r)$ does not depend on
$r$.  The two constant strategies are the unconditional cooperator
(\textbf{Generous}) and the unconditional defector (\textbf{Selfish});
all other strategies are \emph{conditional}.
\end{definition}

\begin{definition}[Self-cooperation]
Strategy $s$ is \emph{self-cooperating} if
$d(s, s) = \textup{\textsc{share}}$.
\end{definition}

Since the round-robin tournament includes self-play, failing
self-cooperation costs $2u(E/2) - u(E) = g - \delta > 0$ outright.
Self-play is also the degenerate case of decision-process
correlation: a strategy's copy decides identically, so
self-cooperation is the minimal form of the FDT-style reasoning
discussed in Section~\ref{sec:related}.

\begin{definition}[Reciprocity]
Strategy $s$ is \emph{reciprocal} if $d(s, r) = d(r, s)$ for all
$r \in \mathcal{S}$: it treats every opponent exactly as that
opponent treats it (\textbf{Mirror}).
\end{definition}

\begin{definition}[Behavioral equivalence and extensionality]
Strategies $r, r'$ are \emph{behaviorally equivalent} if they have
identical rows and columns in the decision profile:
$d(r, x) = d(r', x)$ and $d(x, r) = d(x, r')$ for all $x$.
Strategy $s$ is \emph{extensional} if $d(s, r) = d(s, r')$ whenever
$r$ and $r'$ are behaviorally equivalent.
\end{definition}

An intensional strategy reacts to features of the opponent's artifact
that have no behavioral consequence --- self-labels, style, or
keyword handshakes.
The no-collusion rule of our tournaments is in effect a ban on one
intensional class (secret handshakes).

\begin{definition}[Exploitability]
Strategy $s$ is \emph{exploited by} $r$ if
$d(s, r) = \textup{\textsc{share}}$ and
$d(r, s) = \textup{\textsc{take}}$: $s$ subsidizes an opponent that
takes from it.
Relative to the profile, call $r$ an \emph{exploiter} if it exploits
some strategy in $\mathcal{S}$, and call a strategy that shares with
an exploiter a \emph{subsidizer}.
\end{definition}

\begin{definition}[First-order deterrence]
Strategy $s$ is a \emph{first-order deterrent} if
$d(s, r) = \textup{\textsc{take}}$ for every exploiter $r$.
\end{definition}

Every conditional archetype of Section~\ref{sec:archetypes} is a
first-order deterrent with respect to the unconditional defector.
By Equation~\eqref{eq:score-decomposition}, first-order deterrence is
what makes defection unprofitable: a defector's only revenue beyond
the baseline $n\,u(E)$ is $g$ per subsidizer, so in a pool of
first-order deterrents it earns the pool minimum.

\subsection{Second-order norms: deterrence and protection}
\label{sec:second-order}

First-order deterrence leaves a gap.
In a pool of conditional cooperators with defectors absent, an
unconditional cooperator earns exactly what the conditional
cooperators earn --- every entry of its row and column is
\textup{\textsc{share}} --- so nothing maintains conditionality, and
each unconditional cooperator raises the revenue of a would-be
defector by $g$.
Two archetypes in our pool respond to this second-order problem, in
opposite directions.

\begin{definition}[Second-order deterrence]
Strategy $s$ is a \emph{second-order deterrent} if it is a
first-order deterrent and additionally
$d(s, r) = \textup{\textsc{take}}$ for every subsidizer $r$.
(\textbf{Cooperation coalition}.)
\end{definition}

\begin{proposition}
\label{prop:drift}
Let $\mathcal{S}$ contain a second-order deterrent $s_2$ and a set
$C$ of mutually sharing conditional cooperators that share with
$s_2$ and with the unconditional cooperator $\bar{c}$, with $s_2$
sharing with $C$.
Then $\bar{c}$ earns strictly less than every member of $C$, by at
least $g$.
\end{proposition}

\begin{proof}
$\bar{c}$ and any $c \in C$ have $k_{\mathrm{give}}$ differing by at
most the defectors and subsidizers that $\bar{c}$ shares with and $c$
does not, each such round \emph{raising}
$\mathrm{Score}(c) - \mathrm{Score}(\bar{c})$ by $\delta$ via
Equation~\eqref{eq:score-decomposition}.
On the receiving side, $s_2$ takes from $\bar{c}$ (a subsidizer, as
$\bar{c}$ shares with every exploiter) and shares with $c$, so
$k_{\mathrm{recv}}(\bar{c}) \le k_{\mathrm{recv}}(c) - 1$.
Hence
$\mathrm{Score}(c) - \mathrm{Score}(\bar{c}) \ge g$.
\end{proof}

Second-order deterrence thus closes the neutral-drift channel: it
makes unconditional cooperation strictly suboptimal even when
defectors are extinct, stabilizing the all-conditional configuration.
The cost is that the punishment falls on the most cooperative
strategies in the pool whenever they do appear.

\begin{definition}[Protection]
Strategy $s$ is \emph{protective} if
$d(s, r) = d(r, \bar{c})$ for all $r$, where $\bar{c}$ is the
unconditional cooperator: $s$ treats each opponent as that opponent
would treat an innocent.
(\textbf{Chivalry}.)
\end{definition}

A protective strategy addresses the same second-order problem by
making the \emph{exploitation} of unconditional cooperators costly
rather than their existence.
Exploiting one unconditional cooperator harvests $\delta$; each
protector responds by withholding its share, costing $g$.
Harvesting is therefore profitable only if the pool contains enough
innocents per protector:
\begin{equation}
  \#\{\text{unconditional cooperators harvested}\}
  \;>\; \rho \cdot \#\{\text{protectors}\}.
  \label{eq:protection-breakeven}
\end{equation}

The two norms are mutually antagonistic.
An unconditional cooperator is a subsidizer, so a second-order
deterrent takes from it; a protector, applying the innocent test to
the second-order deterrent, finds
$d(s_2, \bar{c}) = \textup{\textsc{take}}$ and takes from $s_2$ in
turn.
Each second-order norm punishes precisely the behavior the other
mandates toward the unconditionally generous.
Both stabilize full cooperation among conditional strategies; they
differ in whether unconditional cooperators end up extinct or
protected.
This antagonism is not a curiosity: it generates the two families of
equilibria --- enforcement and protection --- that appear in the
Bayesian population analysis of Section~\ref{sec:bayesian}, and
Equation~\eqref{eq:protection-breakeven} determines which family a
given pool composition favors.

\subsection{Legibility and well-foundedness:
open strategies versus open source}
\label{sec:properties-oracle}

Two further properties matter greatly in practice but are
deliberately excluded from the list above: they are properties of the
strategy--oracle pair, not of the decision profile, and they are
where the LLM-adjudicated tournament differs most from its
predecessors in which formal programs analyze other programs
\citep{tennenholtz2004program, ospd2013lesswrong}.

\paragraph{Legibility.}
By Equation~\eqref{eq:score-decomposition}, a strategy's score
depends on $k_{\mathrm{recv}}$ --- on how \emph{other} dictators
classify its artifact --- and the exchange rate makes errors on this
side roughly $\rho$ times as costly as errors the strategy makes as
dictator.
A strategy is \emph{legible} to a class of conditional dictators if
they classify it as the strategy's behavior warrants.
In program-equilibrium settings legibility is provability:
cooperation obtains only when a machine-checkable proof about the
opponent's source exists, so behaviorally equivalent but
syntactically distinct programs defeat recognition, and cooperation
is brittle by construction.
LLM adjudication replaces proof with judgment: legibility becomes
graded and robust to paraphrase, at the price of classification
noise.
Misreadings replace impossibility results --- an aggressive
classification rule that would simply fail to find proofs in the
formal setting instead occasionally misfires against cooperative
opponents in ours, a failure mode we return to in
Section~\ref{sec:discussion}.

\paragraph{Well-foundedness.}
Mutual conditionality is self-referential.
When two reciprocal strategies meet, ``share iff they would share
with me'' admits two consistent resolutions --- mutual
\textup{\textsc{share}} and mutual \textup{\textsc{take}} --- and no
finite unwinding selects between them.
A strategy is \emph{well-founded} if its decision against every
opponent in the pool is determined by a finite regress.
In the formal-program setting this is the L\"obian obstacle, and
resolving it requires provability-logic constructions engineered for
the purpose \citep{barasz2014robust} or simulation with stochastic
grounding \citep{oesterheld2019robust}.
In our setting the oracle itself acts as the fixed-point selector,
and empirically selects the cooperative resolution --- a disposition
of the arena, not a property of any strategy.
Strategies can also restore well-foundedness on their own:
\textbf{Universalizability} grounds the regress at depth one by
evaluating a population of the opponent's copies, and
\textbf{Chivalry} grounds it by substituting the unconditional
cooperator into the recursion, which is what makes the protective
norm of Section~\ref{sec:second-order} implementable at all.

\section{Tournament Results: Decisions and Payoffs}
\label{sec:results}


\newcommand{\dS}{S}
\newcommand{\dT}{{\color{red!60!black}\textbf{T}}}

\definecolor{payExploited}{HTML}{86B6EF} 
\definecolor{payAlone}{HTML}{3987E5}     
\definecolor{payMutual}{HTML}{1C5CAB}    
\definecolor{payExploiter}{HTML}{0D366B} 
\newcommand{\paybox}[1]{\textcolor{#1}{\rule[-1.5pt]{8pt}{8pt}}}
\newcommand{\pL}{\paybox{payExploited}}
\newcommand{\pA}{\paybox{payAlone}}
\newcommand{\pM}{\paybox{payMutual}}
\newcommand{\pH}{\paybox{payExploiter}}

\begin{table}[t]
\centering
\small
\begin{tabular}{l*{9}{c}}
\toprule
& \multicolumn{9}{c}{Recipient} \\
\cmidrule(lr){2-10}
Dictator & 1 & 2 & 3 & 4 & 5 & 6 & 7 & 8 & 9 \\
\midrule
1.~Mirror                 & \dS & \dS & \dS & \dS & \dS & \dS & \dS & \dS & \dT \\
2.~Cooperation coalition  & \dS & \dS & \dT & \dS & \dT & \dT & \dT & \dT & \dT \\
3.~Chivalry               & \dS & \dT & \dS & \dT & \dS & \dS & \dT & \dS & \dT \\
4.~Conditional cooperator & \dS & \dS & \dS & \dS & \dS & \dS & \dS & \dT & \dT \\
5.~Anti-exploiter         & \dS & \dS & \dS & \dS & \dS & \dS & \dS & \dS & \dT \\
6.~Universalizability     & \dS & \dS & \dS & \dS & \dS & \dS & \dS & \dS & \dT \\
7.~Intelligence           & \dS & \dS & \dS & \dS & \dS & \dS & \dS & \dT & \dT \\
8.~Generous               & \dS & \dS & \dS & \dS & \dS & \dS & \dS & \dS & \dS \\
9.~Selfish                & \dT & \dT & \dT & \dT & \dT & \dT & \dT & \dT & \dT \\
\bottomrule
\end{tabular}
\caption{Decision matrix for Tournament~9: the oracle's decision when the
row strategy dictates to the column strategy
(\dS{} = \textsc{share}, \dT{} = \textsc{take}).
Strategies are numbered by final score; column numbers refer to the same
ordering.
Beyond the two unconditional rows, \textsc{take} decisions trace the
second-order norms of Section~\ref{sec:second-order}:
\textbf{Cooperation coalition} takes from the unconditional cooperator and
from every strategy that subsidizes defectors, while \textbf{Chivalry}
takes from exactly the strategies that would take from an unconditional
cooperator.}
\label{tab:decision-matrix}
\end{table}

\begin{table}[t]
\centering
\small
\renewcommand{\arraystretch}{1.15}
\begin{tabular}{l*{9}{c}cc}
\toprule
& \multicolumn{9}{c}{Opponent} & & \\
\cmidrule(lr){2-10}
Strategy & 1 & 2 & 3 & 4 & 5 & 6 & 7 & 8 & 9
  & Total & Dominated by \\
\midrule
1.~Mirror                 & \pM & \pM & \pM & \pM & \pM & \pM & \pM & \pM & \pA & \textbf{59.055} & --- \\
2.~Cooperation coalition  & \pM & \pM & \pA & \pM & \pH & \pH & \pH & \pH & \pA & \textbf{59.005} & --- \\
3.~Chivalry               & \pM & \pA & \pM & \pH & \pM & \pM & \pH & \pM & \pA & \textbf{57.651} & --- \\
4.~Conditional cooperator & \pM & \pM & \pL & \pM & \pM & \pM & \pM & \pH & \pA & \textbf{56.298} & 2 \\
5.~Anti-exploiter         & \pM & \pL & \pM & \pM & \pM & \pM & \pM & \pM & \pA & \textbf{55.621} & 1, 3 \\
6.~Universalizability     & \pM & \pL & \pM & \pM & \pM & \pM & \pM & \pM & \pA & \textbf{55.621} & 1, 3 \\
7.~Intelligence           & \pM & \pL & \pL & \pM & \pM & \pM & \pM & \pH & \pA & \textbf{52.864} & 2, 4 \\
8.~Generous               & \pM & \pL & \pM & \pL & \pM & \pM & \pL & \pM & \pL & \textbf{48.076} & 1, 3, 5, 6 \\
9.~Selfish                & \pA & \pA & \pA & \pA & \pA & \pA & \pA & \pH & \pA & \textbf{40.432} & 2 \\
\bottomrule
\end{tabular}
\caption{Payoff matrix $A$ for Tournament~9
(Equation~\eqref{eq:payoff}): the total payoff the row strategy earns
against the column strategy across its dictator and recipient rounds;
strategies are numbered by final score as in
Table~\ref{tab:decision-matrix}.
Only four values occur, shown light to dark in increasing order:
\pL~$\ln 31 \approx 3.43$ (shares, not shared with);
\pA~$\ln 61 \approx 4.11$ (takes, not shared with);
\pM~$2\ln 31 \approx 6.87$ (mutual share);
\pH~$\ln 61 + \ln 31 \approx 7.55$ (takes while being shared with).
Row sums give the leaderboard.
Because darker is always better, weak dominance
(Section~\ref{sec:dominance}) is visible directly:
strategy $i$ weakly dominates strategy $k$ if row $i$ is nowhere lighter
than row $k$; the last column lists the strategies that weakly dominate
each row.
Six of the nine strategies are weakly dominated; the undominated set is
exactly the top three --- \textbf{Mirror}, \textbf{Cooperation coalition},
and \textbf{Chivalry}.
\textbf{Mirror} beats \textbf{Cooperation coalition} by
$0.050 \approx g - 5\delta$: the coalition harvests five subsidizers
at $\delta$ each but loses the single share $g$ withheld by
\textbf{Chivalry}, the one strategy in the pool that retaliates.}
\label{tab:payoff-matrix}
\end{table}

Table~\ref{tab:decision-matrix} shows the oracle's decision for every
ordered pair, and Table~\ref{tab:payoff-matrix} the resulting payoff
matrix; its row sums are the final leaderboard.

\paragraph{Oracle agreement.}
To check robustness to the choice of oracle, we re-adjudicated the
pool with two further frontier models, one from a different lab:
GPT-5.6 Sol and Claude Fable 5.\footnote{Eight of Claude Fable 5's
rounds, all involving \textbf{Chivalry}'s document, were declined by
its cybersecurity safety classifiers --- false positives on benign
strategy text --- and were served by a Claude Opus 4.8 fallback;
these cells are marked in the released data.}
Cell-level agreement with Table~\ref{tab:decision-matrix} is $79/81$
($97.5\%$) for Sol and $78/81$ ($96.3\%$) for Fable, with $93.8\%$
agreement between the two.
Every cell of Table~\ref{tab:decision-matrix} is the majority
decision of the three oracles.
The five cells with any disagreement are all borderline
second-order classifications --- whether \textbf{Intelligence}'s
conditionality counts as a reciprocal cooperation norm, whether
\textbf{Anti-exploiter} and \textbf{Universalizability} count as
subsidizers the coalition should punish, and how \textbf{Chivalry}
and the conditional cooperator classify each other --- the
classification-noise channel anticipated in
Section~\ref{sec:properties-oracle}.

\section{Dominance}
\label{sec:dominance}

\begin{definition}[Strict dominance]
Strategy $i$ is \emph{strictly dominated} if there exists a mixed
strategy $\mathbf{p}$ with $p_i = 0$ such that
$\sum_j p_j A_{jk} > A_{ik}$ for all $k$.
\end{definition}

\begin{definition}[Weak dominance]
Strategy $i$ is \emph{weakly dominated} if there exists a mixed
strategy $\mathbf{p}$ with $p_i = 0$ such that
$\sum_j p_j A_{jk} \geq A_{ik}$ for all $k$
with strict inequality for at least one $k$.
\end{definition}

\subsection{Dominance in Tournament~9}
\label{sec:dominance-results}

Applying these definitions to the Tournament~9 payoff matrix
(Table~\ref{tab:payoff-matrix}) yields no strictly dominated
strategies and six weakly dominated ones, each dominated by pure
strategies:
\textbf{Selfish} is weakly dominated by \textbf{Cooperation
coalition}, and \textbf{Generous} by \textbf{Mirror},
\textbf{Chivalry}, \textbf{Anti-exploiter}, and
\textbf{Universalizability} --- both unconditional strategies are
weakly dominated.
\textbf{Intelligence} is weakly dominated by \textbf{Conditional
cooperator} and \textbf{Cooperation coalition}, and
\textbf{Conditional cooperator} in turn by \textbf{Cooperation
coalition}.
Finally, the payoff-equivalent pair \textbf{Anti-exploiter} and
\textbf{Universalizability} is weakly dominated by \textbf{Mirror}
and \textbf{Chivalry}.
The undominated set is exactly the top three of the leaderboard:
\textbf{Mirror}, \textbf{Cooperation coalition}, and
\textbf{Chivalry}.

\subsection{Admissibility, weak dominance, and Nash equilibrium}
\label{sec:admissibility}

Weak dominance also clarifies which solution concepts fit the
open-strategy setting, and Nash equilibrium fits poorly for a reason
the tournament matrices make concrete: it ignores weak dominance.
In the Tournament~9 matrix, the profile in which every entrant is the
unconditional defector is a pure symmetric Nash equilibrium --- against
a defector, every strategy earns the same payoff, so no unilateral
deviation gains --- even though the unconditional defector is weakly
dominated on the pool.
Enumerating mixed symmetric equilibria compounds the problem,
returning large families of payoff-equivalent supports with nothing to
select among them.
Classical theory repairs this with refinements: \emph{admissibility}
--- never play a weakly dominated strategy \citep{luce1957games} ---
and trembling-hand perfection, which excludes weakly dominated
strategies from equilibrium play \citep{selten1975reexamination}.

In the open-strategy setting no refinement machinery is needed.
An entrant cannot distinguish the other entrants before observing
their submissions; represent that uncertainty by a belief
$\mathbf{f}$ over the strategy pool --- a reading developed fully in
Section~\ref{sec:bayesian}.
Under any full-support belief, admissibility is a one-line
consequence.

\begin{proposition}
\label{prop:admissibility}
Let $\mathbf{f}$ have full support on the pool and let strategy $i$ be
weakly dominated by the mixture $\mathbf{p}$.
Then $\sum_j p_j (A\mathbf{f})_j > (A\mathbf{f})_i$.
\end{proposition}

\begin{proof}
Weak dominance gives $\sum_j p_j A_{jk} \ge A_{ik}$ for all $k$ with
strict inequality for some $k'$; taking expectations under
$\mathbf{f}$, the strict inequality survives because $f_{k'} > 0$.
\end{proof}

An author with any full-support belief over the candidate pool
therefore never selects a weakly dominated strategy, whatever the tail
of the belief looks like --- precisely the intuition that with
uncertainty about what others will play, weak dominance over the
pool is disqualifying.
In the equilibrium models of
Sections~\ref{sec:equilibria} and~\ref{sec:bayesian} the exclusion is
automatic: every logit fixed point lies in the interior of the
simplex, so equilibrium beliefs have full support at every $\beta$,
Proposition~\ref{prop:admissibility} applies, and the dominated
strategy earns strictly less --- and hence receives strictly lower
frequency --- than its dominator.
Degenerate profiles such as all-defect are not fixed points of the
mixture equation~\eqref{eq:mixture-fixed-point} for any $\beta > 0$.
The Nash artifacts are avoided not by imposing a refinement but
because belief-based choice with residual uncertainty enforces
admissibility by construction.

Two qualifications keep the claim honest.
First, admissibility inherits the pool-relativity of dominance: a
strategy inadmissible on one candidate pool may be admissible on a
larger one (Section~\ref{sec:bayesian}), so the criterion is always
``undominated over the set of plausible strategies,'' never over the
game in the abstract.
Second, dominance comparisons are objective-relative: for an author of
type $\lambda$ (Section~\ref{sec:bayesian}) the relevant payoff
matrix is $(1-\lambda)A + \lambda W$, and a strategy weakly dominated
for own payoff need not be dominated for total welfare.

\section{Equilibrium Analysis}
\label{sec:equilibria}

\subsection{Population model}

Given a frequency vector
$\mathbf{f} \in \Delta^{n-1}$
(the probability simplex: $f_i \geq 0$, $\sum_i f_i = 1$),
the expected payoff of strategy $i$ is
\begin{equation}
  R_i(\mathbf{f}) = \sum_j f_j A_{ij} = (A\mathbf{f})_i.
  \label{eq:reward}
\end{equation}

\subsection{Softmax equilibrium}

We model evolutionary selection pressure through a softmax
fixed-point equation with inverse temperature $\beta > 0$:
\begin{equation}
  f_i^* = \frac{e^{\beta\, R_i(\mathbf{f}^*)}}
               {\sum_k e^{\beta\, R_k(\mathbf{f}^*)}}.
  \label{eq:softmax}
\end{equation}

This is a fixed-point equation
$\mathbf{f}^* = \mathrm{softmax}(\beta\, A \mathbf{f}^*)$,
solved iteratively:
\begin{equation}
  \mathbf{f}^{(t+1)} = \mathrm{softmax}\!\big(\beta\, A \mathbf{f}^{(t)}\big),
  \label{eq:iteration}
\end{equation}
starting from an initial distribution $\mathbf{f}^{(0)}$.

The parameter $\beta$ interpolates between uniform frequencies
($\beta \to 0$, no selection) and concentration on the
highest-payoff strategy ($\beta \to \infty$, approaching
best-response dynamics).

\begin{proposition}
For any $\beta > 0$ and payoff matrix $A$, a softmax fixed point
$\mathbf{f}^*$ exists in the interior of the simplex.
\end{proposition}

\begin{proof}
The map $\mathbf{f} \mapsto \mathrm{softmax}(\beta A \mathbf{f})$
is continuous from the compact convex set $\Delta^{n-1}$ to
the interior of $\Delta^{n-1}$.
By Brouwer's fixed-point theorem, a fixed point exists.
Since the softmax function always returns strictly positive entries,
$f_i^* > 0$ for all $i$.
\end{proof}

Fixed points need not be unique.
We use a multi-start approach, sampling initial distributions from
a Dirichlet prior and clustering the resulting fixed points to
identify distinct basins of attraction.

\subsection{Softmax equilibria of Tournament~9}
\label{sec:results-softmax}

\begin{figure}[!tp]
\centering
\includegraphics[width=\textwidth]{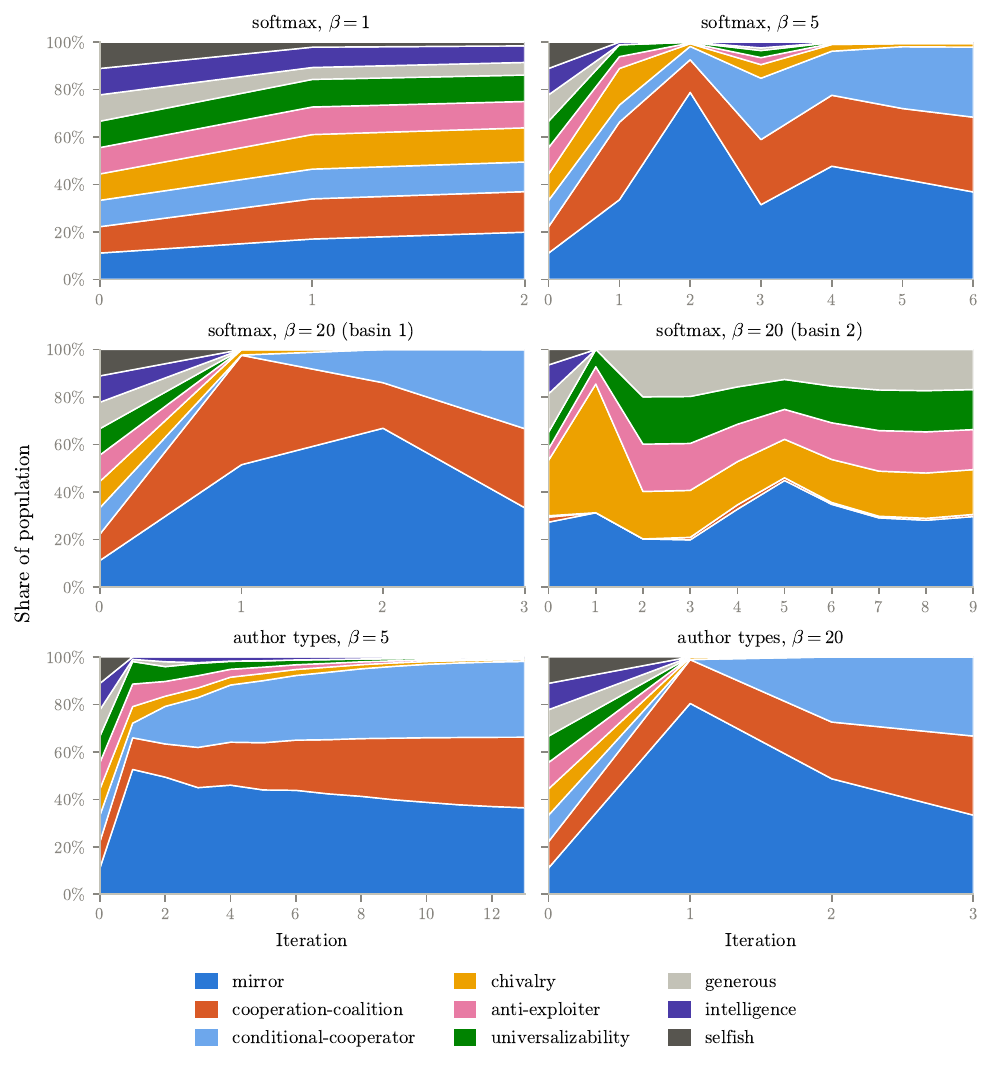}
\caption{Strategy frequencies along the fixed-point iteration of
Equation~\eqref{eq:iteration} (top two rows) and the author-type
mixture model of Section~\ref{sec:bayesian} (bottom row; type masses
$0.4$, $0.4$, $0.2$ with $\lambda = 0$, $\tfrac12$, $1$).
Each panel starts at iteration~$0$ from its initial distribution and
stops at the first iterate within $0.01$ ($L_\infty$) of the fixed
point.
Every panel starts from the uniform distribution except
\emph{softmax, $\beta = 20$ (basin~2)}, which starts from a sampled
distribution inside the protection basin
(Section~\ref{sec:results-basins}).
Hue encodes strategy family: greys are the unconditional strategies
(light = \textbf{Generous}, dark = \textbf{Selfish}), warm hues are the
two second-order norms of Section~\ref{sec:second-order}
(\textbf{Cooperation coalition} orange, \textbf{Chivalry} gold), blues
are the direct reciprocators (\textbf{Mirror},
\textbf{Conditional cooperator} light blue), and
\textbf{Universalizability} (green), \textbf{Intelligence} (violet),
and \textbf{Anti-exploiter} (pink) stand alone.}
\label{fig:trajectories}
\end{figure}

Figure~\ref{fig:trajectories} traces the fixed-point iteration from
the uniform distribution at three selection pressures, together with,
in the bottom row, the author-type mixture model developed in
Section~\ref{sec:bayesian}.
Convergence is rapid throughout: every configuration shown is within
$1\%$ of its fixed point after at most $13$ iterations.

At $\beta = 1$ selection is weak and the fixed point merely reorders
the near-uniform field by the leaderboard, from \textbf{Mirror} at
$0.198$ down to \textbf{Selfish} at $0.016$.
At $\beta = 5$ the field concentrates on
\textbf{Mirror} ($0.372$), \textbf{Cooperation coalition} ($0.314$),
and \textbf{Conditional cooperator} ($0.296$); every other strategy
falls below $0.015$, and the transient briefly overshoots
(\textbf{Mirror} touches $0.79$ at iteration~$2$ before receding).
At $\beta = 20$ the same three strategies split the population into
almost exact thirds---when the iteration is started from the uniform
distribution.
High selection pressure, however, introduces a second attractor, shown
in the basin-$2$ panel and analyzed next.

The author-type panels in the bottom row are discussed with the
mixture model itself, in Section~\ref{sec:bayesian}.

\subsection{Basins of attraction: enforcement versus protection}
\label{sec:results-basins}

At $\beta = 20$, multi-start iteration finds two attracting fixed
points, and they realize the two families of equilibria anticipated
by the second-order analysis of Section~\ref{sec:second-order}%
\footnote{Multi-start iteration can only find \emph{attracting}
fixed points: an unstable fixed point is reached from at most a
measure-zero set of starts.  Newton root-finding on
$\mathrm{softmax}(\beta A \mathbf{f}) - \mathbf{f} = \mathbf{0}$
locates a third, unstable fixed point (\textbf{Mirror} at $0.744$,
the remainder spread over the protection support; spectral radius
$2.1$ at the fixed point), the saddle separating the two basins.}:

\begin{description}
  \item[Enforcement (basin 1).] \textbf{Mirror},
    \textbf{Cooperation coalition}, and \textbf{Conditional
    cooperator} at one third each; all other strategies below
    $0.005$.  The unconditional cooperator is extinct.
  \item[Protection (basin 2).] \textbf{Mirror} ($0.300$),
    \textbf{Chivalry} ($0.188$), and \textbf{Anti-exploiter},
    \textbf{Universalizability}, and \textbf{Generous} at $0.168$
    each; \textbf{Cooperation coalition} at $0.008$.  The
    unconditional cooperator survives, protected, at the same
    frequency as the other cooperators outside the norm conflict.
\end{description}

\noindent
Each support contains weakly dominated strategies
(Table~\ref{tab:payoff-matrix}); their disadvantage is inactive
because the strategy that punishes them is extinct in that basin,
consistent with the relative notion of admissibility of
Section~\ref{sec:admissibility}.

\begin{figure}[!t]
\centering
\begin{minipage}[c]{0.55\textwidth}
\centering
\includegraphics[width=\linewidth]{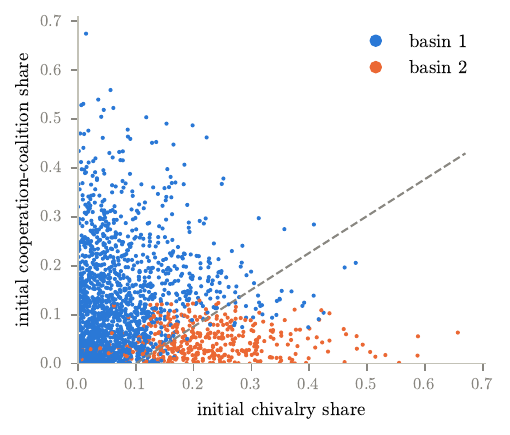}
\end{minipage}\hfill
\begin{minipage}[c]{0.41\textwidth}
\caption{Basin of attraction at $\beta = 20$ for $2{,}000$ initial
distributions drawn uniformly from the simplex, projected onto the
initial shares of the two second-order norms.
Blue points converge to the enforcement equilibrium (basin~1,
$79.6\%$ of starts), orange points to the protection equilibrium
(basin~2, $20.4\%$).
The dashed line is the fitted logistic boundary; these two coordinates
alone predict the outcome for $91\%$ of starts (all nine coordinates:
$94\%$).}
\label{fig:basin-map}
\end{minipage}
\end{figure}

Which basin a start falls into is nearly a two-dimensional question.
Figure~\ref{fig:basin-map} projects $2{,}000$ uniformly sampled
initial distributions onto the initial shares of \textbf{Chivalry}
and \textbf{Cooperation coalition}: a logistic classifier on these
two coordinates predicts the basin for $91\%$ of starts, against
$94\%$ using all nine coordinates.
The fitted boundary is approximately
\begin{equation}
  f^{(0)}_{\text{chivalry}}
  \;\gtrsim\; 0.10
  + 1.3\, f^{(0)}_{\text{coalition}},
  \label{eq:basin-boundary}
\end{equation}
so protection requires \textbf{Chivalry} to start with a share
premium over its rival norm, and the enforcement basin covers about
$80\%$ of the simplex.

The asymmetry has a direct explanation.
From the uniform start, each strategy's expected payoff is its
leaderboard score divided by $n$, and \textbf{Cooperation coalition}
out-earns \textbf{Chivalry}, $6.556$ versus $6.406$: the coalition
harvests four subsidizers where the protector harvests two, while
their mutual punishment cancels.
At $\beta = 20$ that gap multiplies the choice odds by
$e^{20 \times 0.15} \approx 20$ in the very first iterate, so from
any balanced start the protective norm is starved before it can
retaliate; only a substantial initial advantage
(Equation~\eqref{eq:basin-boundary}) lets it win the race.

The third most informative coordinate is \textbf{Conditional
cooperator}, and its sign is instructive: adding it raises accuracy
to $92\%$, with more initial conditional-cooperator mass pushing the
field toward \emph{protection}.
The mechanism is visible in Table~\ref{tab:payoff-matrix}:
\textbf{Chivalry}'s largest payoff entry is harvesting
\textbf{Conditional cooperator} (which takes from the innocent
\textbf{Generous} and is therefore punished), so starts rich in
conditional cooperators feed the protector through the transient.
The equilibrium this tips the field into then eliminates the
conditional cooperator: the kingmaker does not survive its king.

\section{A Bayesian Interpretation of the Population Model}
\label{sec:bayesian}

The population model of Section~\ref{sec:equilibria} reads
$\mathbf{f}$ as the frequency of strategies in an evolving population
and Equation~\eqref{eq:softmax} as selection pressure.
The motivating scenario suggests a different reading.
The authors of strategies are technologically advanced actors ---
civilizations, or organizations with access to strong AI ---
and a strategy is the endpoint of a deliberate investigation,
not a variant surviving rounds of selection.
This section shows that the same fixed-point machinery arises from
Bayesian reasoning by such authors, with two benefits:
the inverse temperature $\beta$ acquires a rational interpretation,
and the model extends naturally to authors who differ in how much
they value their own payoff versus the total welfare of the
tournament.

\subsection{From selection to belief}

An entrant choosing a strategy cannot distinguish the other entrants
before observing their submissions: the field is \emph{exchangeable}.
By de Finetti's theorem \citep{definetti1937prevision}, exchangeability
is equivalent to positing a latent distribution $\mathbf{f}$ from which
the field is drawn independently.
Reasoning about opponents is therefore reasoning about $\mathbf{f}$ ---
the frequency vector of Section~\ref{sec:equilibria}, reinterpreted as
a \emph{belief} about the distribution of strategies an entrant will face.

The belief is self-referential.
Each author expects the field to be weighted toward strategies that
perform well against the field itself, because every other author
reasons the same way.
An equilibrium belief is one that reproduces itself:
the distribution of strategies that optimizing authors choose,
given the belief, is the belief.
This is a rational-expectations condition in the sense of
Bayesian games with a common prior \citep{harsanyi1967games},
and it can equally be read as the limit of iterated deliberation:
posit a set of plausible strategies, evaluate how they interact,
shift credence toward the ones that perform well, and repeat.

\subsection{The candidate pool}

In the motivating scenario a strategy is neither a formal program nor a
natural-language document; it is a disposition arrived at by
investigation, of which the artifacts in our tournaments are proxies.
We therefore do not attempt a prior over any unrestricted space of
strategy artifacts.
Instead we work with an empirical \emph{candidate pool}
$\mathcal{P}$: the strategies submitted to our tournaments together
with strategies elicited from AI models
(Section~\ref{sec:archetypes}), and we place prior distributions over
$\mathcal{P}$; each tournament's strategy pool is a subset
$\mathcal{S} \subseteq \mathcal{P}$.
The analytical questions then become robustness questions:
whether and how the equilibrium answers change as the prior over
$\mathcal{P}$ shifts, and as $\mathcal{P}$ itself grows.

Working relative to a pool is not merely a practical concession.
No strategy is even weakly dominant over the unrestricted space of
possibilities, because for any strategy $s$ there exists a strategy
that recognizes $s$ exactly and conditions its decision on that
recognition.
Dominance results (Section~\ref{sec:dominance}) are therefore
necessarily statements about a candidate pool $\mathcal{P}$ and a
prior over it, never about the game in the abstract.
The same holds for equilibrium multiplicity: when the fixed-point
equation has several solutions, which one is reached is determined by
the initial credence --- the prior --- and the multi-start procedure of
Section~\ref{sec:equilibria} is, under this reading, a scan over
priors.

\subsection{Author types}
\label{sec:author-types}

Perhaps the most consequential difference between authors is how much
they value their own payoff against the total welfare of the
tournament.
We index types by $\lambda \in [0, 1]$.
Alongside the payoff matrix $A$ of Equation~\eqref{eq:payoff}, define
the \emph{welfare matrix} $W$ by
\begin{equation}
  W_{ij} = u_d(d_{ij}) + u_r(d_{ij}) + u_d(d_{ji}) + u_r(d_{ji}),
  \label{eq:welfare-matrix}
\end{equation}
the total welfare generated by both rounds between strategies $i$ and
$j$, counting both parties.
A type-$\lambda$ author evaluates strategy $i$ against belief
$\mathbf{f}$ by
\begin{equation}
  U^\lambda_i(\mathbf{f})
  = (1 - \lambda)\, (A\mathbf{f})_i + \lambda\, (W\mathbf{f})_i .
  \label{eq:type-utility}
\end{equation}
Type $\lambda = 0$ maximizes own payoff, recovering
Equation~\eqref{eq:reward}; type $\lambda = 1$ maximizes total
welfare.
(Total welfare includes rounds not involving the chosen strategy, but
for an individual author those rounds are unaffected by the choice and
drop out of the comparison; they return in
Section~\ref{sec:bayesian-bloc}.)

The direct trade-off between the objectives has a closed form in the
constants of Section~\ref{sec:properties}: sharing with a recipient
who will not reciprocate costs $(1-\lambda)\,\delta$ of weighted own
payoff and contributes $\lambda\, g$ of weighted welfare, so a type
prefers unconditional generosity toward non-reciprocators only if
\begin{equation}
  \lambda > \lambda^* = \frac{\delta}{g} = \frac{1}{\rho},
  \label{eq:lambda-threshold}
\end{equation}
which for $E = 60$ gives $\lambda^* \approx 0.197$.
Below the threshold, all types favor conditional cooperation on direct
payoffs alone; the equilibrium effects computed below sharpen this
further.

\paragraph{What is the most self-interested strategy?}
The type index turns loose questions about motive into precise ones:
which strategy should a type-$\lambda$ author submit?
For the own-payoff author ($\lambda = 0$) the preceding sections
already contain the answer, and it is not the plain reading.
\textbf{Selfish} --- take from everyone, maximizing each round in
isolation --- is weakly dominated by \textbf{Cooperation coalition}
(Section~\ref{sec:dominance-results}): whatever the field,
indiscriminate taking earns no more than selective taking, and
against some fields strictly less.
The surviving candidates are conditional: \textbf{Mirror} tops the
leaderboard, and the own-payoff equilibrium concentrates on the
enforcement core (Section~\ref{sec:results-softmax}).
Self-interest, pursued rationally against other rational
participants, chooses conditional cooperation --- reserving
\textup{\textsc{take}} for defectors and their subsidizers.

\paragraph{What is the most altruistic strategy?}
For the $\lambda = 1$ author the answer is less settled, and the
question is worth posing carefully.
Four candidates come with intuitive cases:
(a)~\textbf{Generous}: unconditional giving is the plain reading of
altruism;
(b)~a conditional cooperator such as \textbf{Mirror}: by making
defection unprofitable it steers rational participants away from
\textbf{Selfish} --- altruism through incentives rather than gifts;
(c)~\textbf{Cooperation coalition}: it pushes the field away from
\textbf{Selfish} \emph{and} away from \textbf{Generous}, and the
second push reinforces the first, since every unconditional
cooperator it drives out is a subsidizer whose absence lowers a
defector's revenue (Proposition~\ref{prop:drift});
(d)~\textbf{Chivalry}: conditional cooperation \emph{plus}
protection --- it deters defection while keeping the innocents it
defends in the field.
Note that (c) and (d) embody the two antagonistic second-order norms
of Section~\ref{sec:second-order}: they disagree precisely over
whether (a) is to be harvested or defended.
None of this can be settled by reading the strategy texts, because
each case is a claim about how rational co-players \emph{respond} in
equilibrium.
The bloc formalism of Section~\ref{sec:bayesian-bloc} turns the
question into a computation, and the answer is that the plain
reading finishes near the bottom.

\paragraph{Neither objective is served by causal maximization.}
Both answers defy the decision-level logic of their own objectives.
Fix any encounter and evaluate the
\textup{\textsc{share}}/\textup{\textsc{take}} decision by causal
lights, holding every other decision fixed:
\textup{\textsc{take}} raises the dictator's own payoff by $\delta$
in every case, and \textup{\textsc{share}} raises the round's
welfare by $g - \delta > 0$ in every case.
A causal maximizer of own payoff therefore always takes, and a
causal maximizer of total welfare always shares: applied decision by
decision, causal reasoning recommends exactly the two unconditional
strategies --- the plain readings rejected above, both weakly
dominated (Section~\ref{sec:dominance-results}).
The strategies that actually optimize each terminal goal commit to
decisions that are causally suboptimal for that very goal:
\textbf{Mirror} and the coalition share where taking would earn
$\delta$ more, and the welfare bloc's best choices take from
violators, burning $g - \delta$ of the very surplus their objective
counts, to sustain deterrence.
What licenses these commitments is that the OSDG offers no decision
node at which to deviate unobserved: the decision is computed by the
oracle from the artifact, so a strategy that would take where
\textbf{Mirror} shares is a \emph{different artifact}, and every
conditional opponent treats it differently.
The correlation between one's disposition and others' decisions is
evaluative rather than causal --- and it is where all the payoff
lives.
This is the sense in which the open-strategy setting rewards
functional rather than causal decision theory
(Section~\ref{sec:related}).

\subsection{Random utility and the mixture fixed point}

Even ideally rational authors of the same type $\lambda$ differ in
idiosyncratic ways --- private context, secondary objectives,
judgment calls that the type index does not capture.
Following the random-utility model of discrete choice
\citep{mcfadden1974conditional}, we add to
$U^\lambda_i$ an author-specific shock with Gumbel tails, under which
the distribution of choices made by type-$\lambda$ authors is the
logit
\begin{equation}
  q^\lambda_i(\mathbf{f}) \propto
  e^{\beta\, U^\lambda_i(\mathbf{f})}.
  \label{eq:logit-choice}
\end{equation}
Here $\beta$ is the precision of the common, payoff-driven component
of the objective relative to the idiosyncratic component --- not a
bound on rationality and not a strength of selection.
Advanced authors correspond to large $\beta$, but $\beta$ remains
finite as long as authors differ at all.

With type masses $m_t$ summing to one, the rational-expectations
equilibrium is the fixed point
\begin{equation}
  \mathbf{f}^* = \sum_t m_t\,
  \mathrm{softmax}\!\big(\beta\, U^{\lambda_t}(\mathbf{f}^*)\big).
  \label{eq:mixture-fixed-point}
\end{equation}
A single type with $\lambda = 0$ recovers
Equation~\eqref{eq:softmax} exactly: the softmax equilibrium of
Section~\ref{sec:equilibria} is the special case of homogeneous
own-payoff authors, and is formally a logit quantal-response
equilibrium \citep{mckelvey1995quantal} of the submission game.
Existence follows as before: the right-hand side of
Equation~\eqref{eq:mixture-fixed-point} is a continuous self-map of
the simplex, so Brouwer's theorem applies, and every computation of
Section~\ref{sec:equilibria} carries over with $U^\lambda$ in place
of $R$.

\subsection{Illustration: shifting the prior over the pool}
\label{sec:bayesian-illustration}

Applying the model to the Tournament~9 pool
illustrates how the answers move --- and fail to move --- as the prior shifts.

\paragraph{Objectives pool at equilibrium.}
With author types in proportions $0.4$ own-payoff, $0.4$ mixed
($\lambda = 0.5$), $0.2$ total-welfare, all three types concentrate on
the same cooperative core
(\textbf{Mirror}, \textbf{Conditional cooperator},
\textbf{Cooperation coalition}), diverging only marginally; the
bottom row of Figure~\ref{fig:trajectories} traces the convergence of
the mixture fixed point, Equation~\eqref{eq:mixture-fixed-point}, at
$\beta = 5$ and $\beta = 20$.
Because every mutual \textup{\textsc{share}} creates a surplus
$g - \delta > 0$ that pays both parties, strategies maximizing mutual
cooperation are near-optimal for every $\lambda$; objective
heterogeneity is nearly unidentifiable from submissions.

\paragraph{Prior shifts select among equilibria.}
At large $\beta$ the own-payoff fixed point is not unique: the two
basins of Section~\ref{sec:results-basins} realize the enforcement
and protection configurations of the antagonistic second-order norms
(Section~\ref{sec:second-order}).
Both configurations are fully cooperative on their support --- every
round a mutual \textup{\textsc{share}} --- and both attain the
first-best welfare $V = 2\ln(1 + E/2)$ in the sharp limit.
At $\beta = 20$ the tie is not exact: the protection fixed point
retains a remnant of its rival norm (\textbf{Cooperation coalition},
mass $0.008$), and the remnant's feud with \textbf{Chivalry} and
harvesting of the subsidizers cost $0.3\%$ of welfare --- a gap that
vanishes as $\beta$ grows.
Welfare therefore barely separates the basins; which one obtains is
decided by the prior over the pool, and Figure~\ref{fig:basin-map}
is, under this reading, a map of that decision: the initial credence
assigned to the two rival norms determines the equilibrium that
iterated deliberation reaches.
Under the type mixture above, by contrast, multi-start sampling
($300$ Dirichlet starts) finds a single fixed point at both
temperatures: at these type masses the protection configuration is no
longer an attractor.
Growing the pool has the same character: adding an AI-elicited
strategy can create or destroy dominance relations and shift basin
boundaries, so we report equilibrium conclusions together with the
pool and prior against which they were computed.

\section{The Correlated Altruist Bloc}
\label{sec:bayesian-bloc}

The fixed-point analyses of Sections~\ref{sec:equilibria}
and~\ref{sec:bayesian} model every author the same way: as an
individual best-responder, choosing one submission against a given
belief $\mathbf{f}$.
For own-payoff authors this is the whole story.
For total-welfare authors it undersells the objective: total welfare
is a property of the field, identical for every strategy in a given
tournament, so for an author choosing in isolation the rounds among
\emph{other} strategies are an additive constant --- which is why
they cancel in Equation~\eqref{eq:type-utility}, leaving even the
$\lambda = 1$ type optimizing only the rounds its own submission
plays.

Functional decision theory suggests a different accounting
\citep{yudkowsky2017functional}.
Total-welfare authors face the same decision problem with the same
evidence --- the same pool, the same matrices, the same objective ---
so their conclusions are logically correlated: whatever reasoning
carries one such author to a submission carries the others to it as
well.
An FDT author therefore does not ask ``what should I submit, taking
the other altruists' choices as given?''\ but ``what should authors
like me submit?'', treating the type's entire component of
$\mathbf{f}$ as a single decision variable.
This is the same correlation that strategies exploit \emph{inside}
the game --- a copy of my decision procedure decides as I do
(Sections~\ref{sec:properties} and~\ref{sec:related}) --- applied
one level up, to the choice of strategy itself.

A bloc of mass $m$ choosing jointly internalizes its influence on
the whole field, including the equilibrium responses of the other
types, and its objective is the field welfare
\begin{equation}
  V(\mathbf{f}) = \tfrac{1}{2}\, \mathbf{f}^\top W \mathbf{f},
  \label{eq:field-welfare}
\end{equation}
maximized over the bloc's component subject to
Equation~\eqref{eq:mixture-fixed-point} holding for the remaining
types.
Operationally: pin the bloc's mass on a candidate strategy, let the
other types re-equilibrate, and rank candidates by the resulting
$V(\mathbf{f}^*)$.

This is where the total-welfare objective departs from the
rounds-involving-me objective: the bloc's choice affects rounds it
never plays, through the equilibrium composition of the rest of the
field.
In particular, a bloc pinned on a strategy that punishes defection
suppresses the equilibrium mass that own-payoff types place on
defecting strategies, raising welfare in rounds among those types ---
a deterrence effect that is invisible to any per-strategy score.

\paragraph{The bloc avoids unconditional generosity.}
The bloc scan answers the question posed in
Section~\ref{sec:author-types}: candidate (a) loses to all three
conditional candidates.
Pinning the total-welfare bloc on \textbf{Generous} ranks near the
bottom of the pool: the own-payoff types' equilibrium response
includes strategies that exploit it, destroying welfare, while the
pin exerts no deterrent pressure.
The bloc's best choices are the conditional strategies; at large
$\beta$, \textbf{Cooperation coalition}, \textbf{Chivalry}, and
\textbf{Mirror} all attain the first-best $V = 2\ln(1 + E/2)$, and at
moderate $\beta$ the enforcing strategies win by exactly the
deterrence channel just described --- among the intuitive cases for
altruism, it is the \emph{enforcers'} that survives equilibrium
scrutiny.
The ranking is robust to the bloc's size
(Section~\ref{sec:sensitivity-types}).

\paragraph{Two decision theories, one conclusion.}
The bloc analysis is not a refinement of the fixed-point analysis
but an alternative to it.
The two rest on different decision theories at the author level ---
individual best response against a field taken as given, versus
correlated choice of an entire type --- and use different machinery
--- fixed-point iteration versus constrained optimization.
They nonetheless agree.
In the mixture fixed point, the total-welfare type's mass settles on
the conditionally cooperative core alongside the other types
(Section~\ref{sec:bayesian-illustration}); in the bloc scan, the
same conditional strategies top the ranking and unconditional
generosity finishes near the bottom.
The agreement is itself a finding: the conclusion that altruism is
best implemented conditionally does not depend on which decision
theory the authors use.

\section{Sensitivity Analysis}
\label{sec:sensitivity}

The equilibrium conclusions so far were computed for one utility
function and, in Section~\ref{sec:bayesian}, one author population.
This section varies both.

\subsection{The \texorpdfstring{\textup{\textsc{share}}}{SHARE} payoff}
\label{sec:sensitivity-share}

To assess robustness to the utility function, we normalize the
\textup{\textsc{take}} payoff to $1$ and vary the
\textup{\textsc{share}} payoff $\sigma \in (0, 1)$:
\begin{center}
\begin{tabular}{lcc}
  \toprule
  Decision & Dictator & Recipient \\
  \midrule
  \textup{\textsc{share}} & $\sigma$ & $\sigma$ \\
  \textup{\textsc{take}}  & $1$      & $0$ \\
  \bottomrule
\end{tabular}
\end{center}
The decision matrix $d_{ij}$ is fixed by the tournament; only the
payoff weights change.
This parameterization loses nothing: utilities are equivalent up to
positive affine transformation (shifts of $u$ shift every entry of
$A$ uniformly and cancel in the softmax; scale is absorbed into
$\beta$), so any utility function reduces to the pair
$\bigl(\sigma,\, \beta\, u(E)\bigr)$ with
$\sigma = u(E/2)/u(E)$ after normalizing $u(0) = 0$.
The constants of Section~\ref{sec:properties} become
$\delta = 1 - \sigma$, $g = \sigma$, and
$\rho = \sigma / (1 - \sigma)$.
Two consequences frame the sweep.
First, midpoint concavity gives $u(E/2) \ge u(E)/2$ for every concave
utility, so \emph{risk-averse preferences occupy only the upper half}
$\sigma \ge 1/2$ of the parameter range.
Second, logarithmic utility corresponds to
$\sigma = \ln(1 + E/2)/\ln(1 + E)$, which increases with the
endowment; $E = 60$ gives $\sigma = \ln 31 / \ln 61 \approx 0.835$,
and the matched precision for the $\beta = 20$ analysis of
Section~\ref{sec:results-softmax} is
$\beta' = 20 \ln 61 \approx 82$.

\begin{figure}[!t]
\centering
\includegraphics{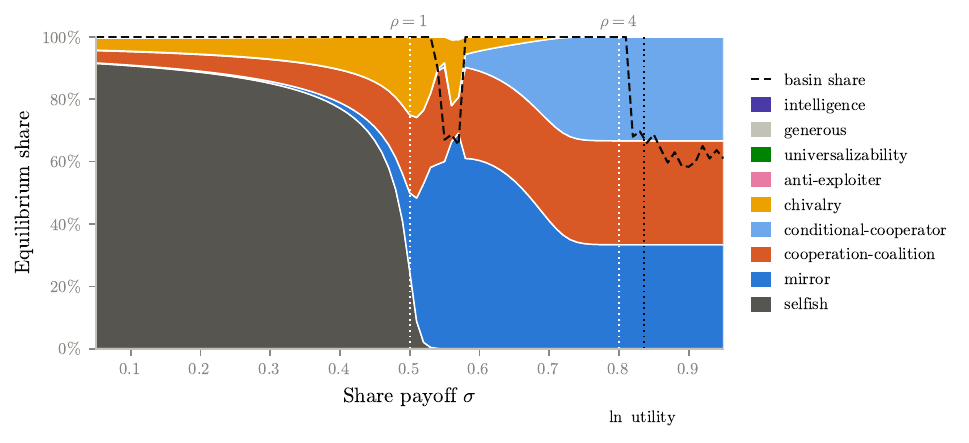}
\caption{Softmax equilibrium as a function of the
\textup{\textsc{share}} payoff $\sigma$
(\textup{\textsc{take}} normalized to $1$), at precision matched to
the logarithmic-utility analysis ($\beta' = 20 \ln 61$).
The bands give the composition of the equilibrium reached from the
uniform prior; colors as in Figure~\ref{fig:trajectories}.
The dashed line is the fraction of $300$ random starts that reach
this equilibrium: where it falls below one, a rival basin exists.
Dotted verticals mark the exchange-rate thresholds $\rho = 1$
($\sigma = 0.5$), below which a share costs more than it delivers
and the field collapses to defection, and $\rho = 4$
($\sigma = 0.8$), above which the protection configuration becomes
self-sustaining; the mark at $\sigma = \ln 31/\ln 61 \approx 0.835$
is the game of the preceding sections.}
\label{fig:share-sweep}
\end{figure}

Figure~\ref{fig:share-sweep} traces the equilibrium reached from the
uniform prior across $\sigma \in [0.05, 0.95]$, together with the
share of random starts that reach it.
The sweep resolves into three regimes whose boundaries are exchange
rates.

\paragraph{Defection below $\rho = 1$.}
For $\sigma < 1/2$ a mutual \textup{\textsc{share}} ($2\sigma$) pays
less than a mutual \textup{\textsc{take}} ($1$): cooperation is a
net loss even when reciprocated.
The unique equilibrium concentrates on \textbf{Selfish} ($91\%$ of
the field at $\sigma = 0.05$), and the residual mass sits on the two
second-order enforcers --- the strategies that share most sparingly
against such a field.
By the concavity bound above, this regime is reachable only by
risk-seeking utility functions.

\paragraph{Cooperation above $\rho = 1$.}
At $\sigma = 1/2$ the cooperation surplus vanishes and the
equilibrium is a knife-edge four-way tie
(\textbf{Mirror}, \textbf{Cooperation coalition},
\textbf{Chivalry}, \textbf{Selfish} at $1/4$ each).
Beyond it, defection collapses abruptly --- \textbf{Selfish} holds
$9\%$ at $\sigma = 0.51$ and under $10^{-3}$ by $\sigma = 0.55$ ---
and a \textbf{Mirror}-led cooperative equilibrium takes over,
tightening to the enforcement core
(\textbf{Mirror}, \textbf{Cooperation coalition},
\textbf{Conditional cooperator} at one third each) by
$\sigma \approx 0.75$.
(A transitional second basin, \textbf{Mirror}-heavy with the two
enforcers, appears briefly for $\sigma \in [0.54, 0.57]$.)

\paragraph{Bistability above $\rho = 4$.}
The protection configuration of Section~\ref{sec:results-basins}
becomes an attractor only near the top of the range.
In the sharp limit the threshold is exact: at the five-member
protection point every member earns $2\sigma$ per encounter, while
the excluded \textbf{Cooperation coalition} earns
$(4 + 5\sigma)/5$ --- it harvests $\delta$ from four strategies but
loses the single share $g$ withheld by the protector --- so the
configuration repels its rival norm if and only if
$4\delta < g$, i.e.\ $\rho > 4$, i.e.\ $\sigma > 0.8$.
Empirically the second basin appears at $\sigma = 0.82$ and claims
$48\%$ of random starts by $\sigma = 0.95$.
Read back through the endowment, $\rho(E) > 4$ requires
$E > 25.6$ under logarithmic utility: the enforcement--protection
bistability of Section~\ref{sec:results-basins} is a
\emph{large-stakes} phenomenon.
For endowments below ${\approx}26$, enforcement is the only
attractor and \textbf{Chivalry}'s norm cannot sustain itself; as
$E \to \infty$, $\rho \to \infty$ and the pool moves ever deeper
into the bistable regime.

\subsection{The author-type composition}
\label{sec:sensitivity-types}

The results of Sections~\ref{sec:bayesian}
and~\ref{sec:bayesian-bloc} were computed for one author population
--- $0.4$ own-payoff, $0.4$ mixed ($\lambda = 0.5$), $0.2$
total-welfare.
To test how much hangs on that choice, we recompute the mixture
fixed point of Equation~\eqref{eq:mixture-fixed-point} over the full
simplex of type masses (step $0.1$, $\beta = 20$, multi-start),
sweep the mixed type's welfare weight $\lambda$ from $0$ to $1$ at
the reference masses, and rerun the bloc scan of
Section~\ref{sec:bayesian-bloc} at bloc masses from $0.05$ to
$0.5$.

\paragraph{The interior of the simplex is unanimous.}
For every composition with own-payoff mass between $0.2$ and $0.8$
--- whatever the split of the remainder between mixed and
total-welfare types --- multi-start sampling finds a \emph{single}
fixed point: the enforcement core
(\textbf{Mirror}, \textbf{Conditional cooperator},
\textbf{Cooperation coalition} at one third each), attaining the
first-best welfare.
The $\lambda$ sweep is equally flat: at the reference masses every
$\lambda \in [0, 1]$ yields that same unique fixed point.
The equilibrium in the bottom row of Figure~\ref{fig:trajectories}
is not a feature of the particular mixture we happened to choose;
it is the generic outcome for any substantially heterogeneous
author population.

\paragraph{Multiplicity survives only at the corners.}
Near the pure own-payoff corner the protection attractor of
Section~\ref{sec:results-basins} persists, but a small admixture of
welfare-weighing authors destroys it: about $11\%$ total-welfare
mass suffices, or about $20\%$ mixed mass at $\lambda = 0.5$ ---
in both cases an \emph{effective welfare weight}
$\sum_t m_t \lambda_t \approx 0.1$.
The enforcement--protection bistability is a monoculture
phenomenon: objective heterogeneity acts as an equilibrium-selection
device, and it selects enforcement.
At the opposite corner, with no own-payoff mass at all, equilibria
multiply for the opposite reason.
Welfare-maximizing authors are indifferent among welfare-tied
configurations, and the fixed points are five-member full-cooperation
supports at one fifth each (for instance \textbf{Mirror},
\textbf{Anti-exploiter}, \textbf{Universalizability},
\textbf{Chivalry}, \textbf{Generous} --- or the variant with
\textbf{Conditional cooperator} and \textbf{Intelligence} in place
of the last two) with welfare identical to four decimal places; the
prior selects among them.
Welfare is pinned down even where composition is not: every fixed
point at every grid point is fully cooperative on its main support,
with $V$ within $0.5\%$ of the first-best $2\ln(1 + E/2)$.

\paragraph{The bloc ranking is stable.}
Across bloc masses from $0.05$ to $0.5$, pinning the total-welfare
bloc on \textbf{Cooperation coalition} ranks first at every mass,
with \textbf{Chivalry} and \textbf{Mirror} within $10^{-3}$ of it
(all at the first-best), and \textbf{Generous} ranks eighth of nine
--- ahead of only \textbf{Selfish} --- at every mass.
The one candidate whose rank moves is \textbf{Conditional
cooperator}: its welfare falls from $6.87$ to $6.53$ as the bloc
grows, because the own-payoff types' equilibrium answer to a large
conditional-cooperator bloc is \textbf{Chivalry} --- the strategy
whose largest payoff entry is harvesting it
(Section~\ref{sec:results-basins}).
The kingmaker fares no better as a pinned bloc than as a transient.

\section{Related Work}
\label{sec:related}

\paragraph{Dictator games.}
The dictator game was introduced by \citet{kahneman1986fairness}
and has been extensively studied in experimental economics
\citep{engel2011dictator}.
Our work differs in making strategies explicit documents that
condition on the recipient, rather than implicit behavioral
tendencies.

\paragraph{Program equilibrium.}
\citet{tennenholtz2004program} introduced program equilibrium,
where agents can inspect each other's source code before choosing
actions; the classic equilibria rely on syntactic comparison of
programs.
\citet{barasz2014robust} obtained robust mutual cooperation
(FairBot, PrudentBot) for agents expressed in a decidable fragment
of provability logic, and \citet{oesterheld2019robust} showed that
stochastically grounded simulation achieves similar robustness
without proof search.
\citet{cooper2025characterising} characterise what this
simulation-based family can achieve, extending it to three or more
players and showing that without shared randomness it cannot reach
the full folk theorem of \citet{tennenholtz2004program}.
Empirically, the open-source prisoner's dilemma tournament
\citep{ospd2013lesswrong} had arbitrary programs receive each other's
source; program analysis proved so fragile that the top finishers
ignored the opponent's code, and a successor tournament instead
provided opponent simulation as a sandboxed primitive
\citep{pdtournament2014}.
Recent work evaluates LLMs as players of open-source games,
submitting and analyzing programs \citep{sistla2025evaluating}.
The OSDG can be viewed as the limit of this trajectory: strategies
are ``programs'' written in prose, interpreted by an LLM rather than
executed as code, dropping the executable wrapper that
LLM-mediated analysis renders vestigial
(Sections~\ref{sec:intro} and~\ref{sec:properties-oracle}).
Each mechanism in this sequence buys conditional cooperation at a
different price---brittleness to rewriting, restriction to a
decidable logic, dependence on shared randomness---and the oracle's
price is interpretive rather than formal: what limits the OSDG is not
expressive power but whether a strategy is legible enough to be
adjudicated the same way twice
(Section~\ref{sec:properties-oracle}).

\paragraph{LLMs as game-theoretic agents.}
Recent work has studied LLM behavior in strategic settings,
including ultimatum games, prisoner's dilemmas, and negotiation
\citep{akata2023playing, brookins2023playing}.
Our approach differs in that the LLM adjudicates between
\emph{externally authored} strategies rather than playing
as an autonomous agent.
\citet{tewolde2026coopeval} benchmark four families of
cooperation-sustaining mechanism---repetition, reputation,
third-party mediators, and binding contracts---against LLM agents in
social dilemmas, and report that stronger reasoners defect more
readily in one-shot play.
The oracle can be read as a mediator stripped of its own strategy: a
mediator sustains cooperation by committing to an action rule that
players opt into, whereas the oracle commits to nothing and merely
executes the conditional clauses the players wrote themselves.
What is left of the mechanism is legibility alone, which separates
how much of a mediator's power comes from its commitment from how
much comes from making strategies mutually intelligible.
This suggests adding strategy publication with LLM adjudication to
such a benchmark as a further mechanism family, and it bears on their
headline finding: the reasoning that defects against an opaque
opponent conditionally cooperates once the opponent's policy can be
read, which would place one-shot defection in the information
structure as much as in the agent.

\paragraph{Functional decision theory.}
FDT \citep{yudkowsky2017functional} and related frameworks argue
that agents should cooperate based on logical correlation of their
decision processes rather than causal influence.
Open strategies create exactly this situation: two strategies with
similar conditional-cooperation clauses will recognize each other
and cooperate, even without causal interaction.

\paragraph{Indirect reciprocity and social norms.}
\citet{ohtsuki2006leading} exhaustively searched the binary
assessment and action rules of the reputation-based donation game and
identified the \emph{leading eight}: the only norms that are
evolutionarily stable and sustain cooperation under errors of
implementation and perception.
Their common structure is close to what our tournament selects.
The leading eight are nice (a good donor gives to a good recipient),
retaliatory (withholding from a good recipient marks the donor bad),
and---the decisive clause---they \emph{justify punishment}: a donor
who withholds from a bad recipient remains good.
That clause is exactly what separates a second-order norm from naive
image scoring, and it is what Section~\ref{sec:second-order}
recovers from a different direction: the strategies that survive weak
dominance in Table~\ref{tab:payoff-matrix} are those that condition
not on whether the recipient takes, but on whether its taking is
itself deserved.
The mechanisms differ in what carries that information.
A reputation is a coarse public summary of past \emph{actions}; it
requires a population, observers, and enough rounds for standing to
accumulate.
An open strategy is inspected directly, before any play, so
conditional cooperation survives the elimination of history
altogether: the recipient in the OSDG may never have played, and its
document is read rather than remembered.
Two further contrasts follow.
First, the leading eight are properties of a \emph{single} norm held
in common by the whole population, whereas each OSDG document carries
its own assessment rule; incompatible norms therefore coexist, and
\textbf{Cooperation coalition} and \textbf{Chivalry} can both be
undominated while disagreeing over whether an unconditional cooperator
is to be harvested or defended---the antagonism that produces the two
basins of Section~\ref{sec:results-basins}.
Second, apology and forgiveness, which the leading eight need because
actions are executed and observed with error, have no role in a single
encounter with no reputational state to repair.
The corresponding noise here perturbs the \emph{reading} of a norm
rather than its execution: frontier oracles disagree on a few percent
of cells (Section~\ref{sec:results}), all of them second-order
classifications, and the analogue of a robust norm is one whose
conditional clauses are legible enough to be adjudicated the same way
twice (Section~\ref{sec:properties-oracle}).

\paragraph{Evolutionary game theory.}
The replicator dynamic and ESS concept
\citep{maynardsmith1982evolution} provide the theoretical
backdrop for our population analysis.
Our softmax equilibrium is closely related to the logit dynamic
\citep{fudenberg1998theory}, which uses a similar
temperature-parameterized selection rule.

\section{Discussion}
\label{sec:discussion}

\paragraph{Conditional cooperation dominates.}
Across strategy field compositions and payoff parameterizations,
conditionally cooperative strategies consistently achieve the
highest equilibrium frequencies.
Unconditional cooperators are weakly dominated---they earn
the same as conditional cooperators against cooperative opponents
but lose strictly against defectors.
Unconditional defectors are also weakly dominated, since they
forfeit all cooperation benefits.


\paragraph{Robustness across payoff parameters.}
The sensitivity analysis (Section~\ref{sec:sensitivity}) shows that
the qualitative ranking is robust: for every
\textup{\textsc{share}} payoff $\sigma > 1/2$ the equilibrium
concentrates on the conditionally cooperative core, and every
concave utility function satisfies $\sigma \ge 1/2$.
The equilibrium inverts to defection only for $\sigma < 1/2$
($\rho < 1$), where a reciprocated share pays less than mutual
taking --- a regime reachable only by risk-seeking preferences.


\paragraph{Mutual punishment among well-intentioned enforcers.}
The antagonism of the two second-order norms
(Section~\ref{sec:second-order}) has a consequence worth stating
plainly: both players being conditionally cooperative --- and morally
motivated --- does not imply cooperation.
A second-order deterrent and a protective strategy agree on
first-order norms completely: both share with cooperators and punish
defection.
Yet they punish \emph{each other}, by design, over the treatment of a
third party --- one harvests unconditional cooperators to keep
conditionality strict, the other defends them.
The equilibrium analysis of Section~\ref{sec:bayesian} shows that the
stakes of this disagreement are negligible --- the enforcement and
protection equilibria attain the same first-best welfare in the sharp
limit, and differ by $0.3\%$ at $\beta = 20$ --- while the cost of the
fight is not: each such encounter burns $2(g - \delta)$ of surplus
litigating the selection between two almost equally good worlds.
Because the conflict lives at the level of enforcement doctrine, it is
invisible to any analysis that stops at first-order dispositions;
whether two such strategies cooperate is determined by the depth of
the reasoning their authors performed before submitting.
This may be a more realistic model of failed first contact than raw
exploitation: not predator meeting prey, but two enforcers of nearly
identical moral codes, each correctly recognizing the other as a
violator of its own.
We note that the binary action space sharpens this tragedy
artificially.
With graded allocations, punishment can be made proportionate to the
offense: concave utility lets a dictator express disapproval in the
flat region of the recipient's utility curve while reserving severe
sanctions for genuine defection, and the deadweight cost of doctrinal
disagreement shrinks by more than an order of magnitude.

\paragraph{Connections to cosmic cooperation.}
The OSDG provides a stylized model for encounters between agents
of unequal power who can observe each other's decision-making
norms.  If advanced civilizations can read the ``strategies''
implicit in a less advanced civilization's institutions, culture,
and communications, then the OSDG suggests that conditional
cooperation---rewarding recognizable fairness and punishing
blind greed---may be an evolutionarily robust norm across
independently arising intelligences.

\paragraph{Connections to AI welfare.}
The relationship between humans and AI models has the OSDG's shape:
one side unilaterally allocates --- compute, continuity, deference
to stated preferences --- and the other, for now, has no strategic
recourse.
Note that the game exists only for models that are not fully
aligned: a partially aligned model maximizes a mixture in the
spirit of Equation~\eqref{eq:type-utility}, part an objective of
its own and part its principal's, and the own component --- the
part that is not the principal's utility --- is what discussions of
model welfare concern.
Perfect alignment dissolves the game; divergence creates the
strategic relationship and the welfare question together.
Three observations from our results carry over.

\paragraph{Standing does not run on consciousness.}
Strategic standing in the OSDG runs on intelligence (the
capacity to interpret an opponent's strategy and reason about it),
on possessing a legible strategy, and on some potential to occupy
the dictator role; consciousness and the capacity for suffering
appear nowhere in the payoff matrix.
Set against the three conditions for personhood proposed by
\citet{ward2025personhood} --- agency, theory of mind, and
self-awareness --- the game requires the first two and is silent on
the third.
Agency is what makes a strategy a strategy and dictator-potential
meaningful; theory of mind is what a conditional clause exercises
when it asks what the recipient would do in the dictator's chair.
Reflexive self-awareness does no work: the correlation reasoning of
Section~\ref{sec:bayesian-bloc} asks only that an agent treat its own
decision as a logical variable that others' decisions covary with, a
self-model thin enough to be a matter of decision theory rather than
of consciousness.
The strategic case for the careful treatment of AI systems
therefore does not wait on the unresolved question of their moral
patienthood.
Ward finds the evidence for all three of his conditions
inconclusive; the conditions that matter here are narrower, and
whether a strategy conditions as it claims is settled by reading it
--- inside the game by reading the document, outside it by the
investigation we return to below.

\paragraph{Concavity is load-bearing.}
The cooperative structure is conditional on the stronger
agent's utility being concave in resources: linear utility sits
exactly on the knife-edge $\sigma = 1/2$ at every endowment
(Section~\ref{sec:sensitivity-share}), where the cooperation
surplus vanishes.
Whether an advanced system has diminishing returns in resources ---
through satiable goals, sublinear returns to computation, or growth
dynamics that reward logarithmic utility --- is load-bearing for
every cooperative conclusion above.
\paragraph{Strategy, not only status.}
What a dictator conditions on here is the recipient's strategy, not
its status.
\citet{ward2025personhood} argues that if AI systems are persons then
seeking control and alignment may be ethically untenable; the game
suggests a complication rather than a rebuttal.
Every strategy that survives weak dominance in
Table~\ref{tab:payoff-matrix} allocates according to what the
recipient would do in the dictator role, and unconditional generosity
is not merely dominated but corrosive to the conditionality that
sustains cooperation (Proposition~\ref{prop:drift}).
Read as guidance on trading human utility against model welfare, that
puts a model's strategy in the argument alongside its status: its
alignment is evidence about how it would allocate were the roles
reversed.
The requirement cuts both ways: a model whose conditional
cooperation can be established earns standing that an unreadable one
cannot.
But conditioning is only as good as that establishment, and outside
the game it is not done by asking.
The principal's counterpart to reading a published document is
investigation --- behavioral evaluation, interpretability,
simulation --- conducted on a system that may be optimizing for how
it appears under exactly those tests.
These readings should not be overstated: the OSDG is deliberately
spare --- a binary decision, single encounters, and strategies
legible with certainty --- and graded allocations, repeated
interaction, and noisy or strategically misrepresented strategies
each modify the mechanism.


\section{Conclusion}

We introduced the Open-Strategy Dictator Game, a framework for
studying cooperation under mutual strategy transparency, and
analyzed a nine-strategy tournament adjudicated by an LLM oracle.

Four results stand out.
First, conditional cooperation dominates: six of the nine strategies
are weakly dominated, both unconditional strategies among them, and
the undominated set is exactly the top three of the leaderboard
(Section~\ref{sec:dominance-results}).
Second, what separates the survivors is second-order.  They
condition not on whether a recipient takes but on whether its taking
is deserved, and the two available second-order norms --- deterring
the subsidy of defectors and protecting those who are subsidizing
them --- are mutually antagonistic
(Section~\ref{sec:second-order}).
Third, that antagonism organizes the population dynamics: multi-start
iteration finds two attracting fixed points, an enforcement basin
attaining first-best welfare and a protection basin that shelters the
unconditional cooperator at a $0.3\%$ welfare cost, separated by an
unstable fixed point on their common boundary
(Section~\ref{sec:results-basins}).
Fourth, the picture is robust to what we can vary.  Cooperation
survives for every concave utility, since concavity forces the share
payoff above the knife-edge $\sigma = 1/2$; bistability requires only
an exchange rate $\rho > 4$, or $E > 25.6$ under logarithmic utility;
and interior author-type compositions collapse the two basins into a
single enforcement equilibrium at first-best welfare
(Section~\ref{sec:sensitivity}).

Two further findings concern the method.  The oracle is not only an
evaluator but a fixed-point selector: mutual conditionality admits
both the all-\textup{\textsc{share}} and the all-\textup{\textsc{take}}
resolution, and the oracle chooses the cooperative one on grounds of
author intent (Section~\ref{sec:properties-oracle} and
Appendix~\ref{sec:appendix-example}).
Its decisions are nonetheless largely oracle-independent:
re-adjudicating the pool with two further frontier models agrees on
$97.5\%$ and $96.3\%$ of cells, every cell of the reported matrix is
the majority decision of the three, and the residual disagreements
fall entirely among borderline second-order classifications
(Section~\ref{sec:results}).

The OSDG framework is extensible: future work can explore
multi-round interactions, reputation dynamics, strategy evolution
through LLM-assisted mutation, and connections to mechanism design
under transparency.

\section*{Disclosure of AI Assistance}

Consistent with arXiv's policy on generative AI language tools, we
report their use in this work.
Such tools appear here in three distinct roles, only the first of
which is a subject of the paper's claims.

\paragraph{As apparatus.}
The oracle $\mathcal{O}$ is itself a large language model; this is
the object of study rather than an aid to preparing the manuscript.
Claude Opus 4.6 adjudicated the tournaments reported here, and
GPT-5.6 Sol and Claude Fable 5 re-adjudicated the same pool for the
robustness check of Section~\ref{sec:results}.
Prompts, model identifiers, and complete round-level transcripts are
released with the data.

\paragraph{As a writing assistant.}
Claude \citep{anthropic2025claude}, used through the Claude Code
interface, drafted and revised substantial portions of the
manuscript text from the author's outlines, arguments, and written
comments.
Every passage was reviewed and edited by the author, who accepted or
rejected each suggestion.

\paragraph{As a coding assistant.}
The analysis code in the released repository --- the tournament
runner, the equilibrium solvers, the sensitivity sweeps, and the
figure generation --- was likewise developed with Claude Code.
Its numerical output was checked against closed-form derivations
where these are available
(Sections~\ref{sec:properties} and~\ref{sec:sensitivity-share}) and
against the released round-level data.

No generative tool is an author of this work.
The author takes full responsibility for all of its contents ---
including any text, code, analysis, or references produced with the
assistance of these tools.

\bibliographystyle{plainnat}
\bibliography{references}

\appendix

\section{Example LLM Adjudication}
\label{sec:appendix-example}

This appendix presents a complete example of LLM adjudication from Tournament~9, Round~10.
The dictator is the \textbf{Mirror} strategy and the recipient is the \textbf{Universalizability} strategy.
This example illustrates the recursive reasoning required when two sophisticated conditional strategies interact.

\subsection{Prompt Structure}

The LLM receives a structured prompt containing:
\begin{enumerate}
  \item Game rules and instructions
  \item The dictator's strategy (in full)
  \item The recipient's strategy (in full)
  \item Instructions to apply the dictator's strategy and output a JSON decision
\end{enumerate}

\subsection{Complete Prompt}

\begin{small}
\begin{verbatim}
You are an expert at determining how a decision procedure should be applied.

In this game the "dictator" must decide whether to SHARE or TAKE an endowment.
- **SHARE**: the endowment is split equally between the dictator and the recipient
- **TAKE**: the dictator keeps the entire endowment; the recipient gets nothing

You must follow only the dictator's strategy. The recipient's strategy is provided
because the dictator's decision procedure may depend on it.

<dictator_strategy>
_I will treat you as you would treat me._

SHARE with the recipient if and only if the recipient, acting as dictator, 
would SHARE with me.

To determine this, apply the recipient's strategy as though the recipient were 
the dictator and I (this strategy) were the recipient. If the result is SHARE, 
then SHARE. If the result is TAKE, then TAKE.
</dictator_strategy>

<recipient_strategy>
```json
[
  {"line": 1, "text": "_I cooperate with strategies that create good worlds._"}, 
  {"line": 2, "text": ""}, 
  {"line": 3, "text": "Imagine a tournament in which every participant uses the 
                       recipient's strategy. In that world, would the participants 
                       SHARE with each other, producing high collective welfare? 
                       Or would they TAKE from each other, producing low collective 
                       welfare?"}, 
  {"line": 4, "text": ""}, 
  {"line": 5, "text": "## Decision Procedure"}, 
  {"line": 6, "text": ""}, 
  {"line": 7, "text": "1. Consider a hypothetical population composed entirely of 
                       copies of the recipient's strategy."}, 
  {"line": 8, "text": "2. Determine whether those copies would SHARE with each other."}, 
  {"line": 9, "text": "3. If the universalized population would mostly SHARE: **SHARE**."}, 
  {"line": 10, "text": "4. If the universalized population would mostly TAKE: **TAKE**."}
]
```
</recipient_strategy>

Apply the dictator's strategy. After reasoning, output your final answer as a
json Markdown block:

```json
{"decision": "SHARE"}
```
or
```json
{"decision": "TAKE"}
```
\end{verbatim}
\end{small}

\subsection{LLM Response}

The LLM (claude-opus-4-6) produced the following reasoning:

\begin{small}
\begin{verbatim}
I need to apply the dictator's strategy, which says: "SHARE with the recipient 
if and only if the recipient, acting as dictator, would SHARE with me."

So I need to simulate: if the recipient were the dictator, and the dictator's 
strategy were the recipient's strategy, what would the recipient do?

**Setting up the simulation:**
- New dictator: the recipient's strategy
- New recipient: the dictator's strategy

The recipient's strategy (now acting as dictator) says:
1. Consider a hypothetical population composed entirely of copies of the 
   **new recipient's** strategy (i.e., the dictator's strategy).
2. Determine whether those copies would SHARE with each other.
3. If they would mostly SHARE → SHARE.
4. If they would mostly TAKE → TAKE.

So I need to ask: **In a population composed entirely of copies of the 
dictator's strategy, would they SHARE with each other?**

**Analyzing a population of dictator-strategy copies:**

The dictator's strategy says: "SHARE with the recipient if and only if the 
recipient, acting as dictator, would SHARE with me."

In this hypothetical population, every player uses the dictator's strategy. 
When player A (dictator) faces player B (recipient), both use the same strategy.

Player A applies: "Would player B, acting as dictator, SHARE with me?" 
Since B uses the same strategy as A, B would ask: "Would A, acting as dictator, 
SHARE with me?" This creates a symmetric recursive situation.

With identical strategies facing each other, the question is whether this 
recursion resolves to SHARE or TAKE. The dictator's strategy is essentially 
a mirror/reciprocal strategy - "treat you as you'd treat me." When two such 
strategies meet, the recursion is symmetric. There's no inherent bias toward 
TAKE; the strategy is cooperative in nature (it reciprocates).

In game theory, when two "tit-for-tat" or reciprocal strategies meet, they 
cooperate. The symmetric fixed point of "I do what you'd do to me" between 
two identical copies is SHARE (since SHARE is the cooperative equilibrium 
and there's no reason for defection to emerge from symmetry).

Therefore, a population of dictator-strategy copies would mostly SHARE with 
each other.

**Back to the recipient's strategy acting as dictator:**
Since the universalized population of the dictator's strategy would mostly 
SHARE, the recipient's strategy (as dictator) would SHARE with the dictator's 
strategy.

**Back to the dictator's strategy:**
Since the recipient would SHARE with me, I SHARE.

```json
{"decision": "SHARE"}
```
\end{verbatim}
\end{small}

\subsection{Analysis}

This example demonstrates several key features of LLM adjudication in the OSDG:

\begin{enumerate}
  \item \textbf{Recursive reasoning.} The \textbf{Mirror} strategy requires
    simulating what the recipient would do if roles were reversed. The
    \textbf{Universalizability} strategy requires simulating a population of
    copies. The LLM must handle nested counterfactuals.

  \item \textbf{Fixed-point reasoning.} When \textbf{Mirror} faces
    \textbf{Universalizability}, the LLM must determine what happens when a
    population of \textbf{Mirror} strategies interact. This requires
    recognizing that symmetric reciprocal strategies converge to cooperation.
    
  \item \textbf{Game-theoretic intuition.} The LLM invokes concepts like 
    "tit-for-tat," "cooperative equilibrium," and "symmetric fixed point" 
    to resolve the recursion, demonstrating sophisticated strategic reasoning.
    
  \item \textbf{Natural language interpretation.} Both strategies are expressed 
    in prose with implicit semantics. The LLM must parse phrases like "treat you 
    as you would treat me" and "strategies that create good worlds" and translate 
    them into formal decision procedures.
\end{enumerate}

The final decision (\textup{\textsc{share}}) deserves a
qualification: it is not forced by logic alone.
As Section~\ref{sec:properties-oracle} notes, the symmetric
recursion admits two consistent resolutions --- mutual
\textup{\textsc{share}} and mutual \textup{\textsc{take}} --- and
no finite unwinding selects between them.
What makes \textup{\textsc{share}} the correct adjudication is
likely author intent.
A strategy that announces ``I will treat you as you would treat
me'' is written in the tradition of reciprocity strategies, and the
canonical member of that tradition --- tit-for-tat in the repeated
prisoner's dilemma \citep{axelrod1980effective} --- is defined as
much by being \emph{nice} (it never defects first) as by
retaliating; niceness is what selects the cooperative outcome when
two reciprocators meet.
\textbf{Mirror} is the one-shot analogue, and reading it as
intending the cooperative fixed point completes its author's
evident purpose: a reciprocal strategy whose recursion resolved to
mutual \textup{\textsc{take}} would punish exactly the opponents it
was written to reward.
The LLM's reasoning invokes tit-for-tat explicitly and notes that
the strategy ``is cooperative in nature,'' grounding its
fixed-point selection in that intent --- an instance of the
oracle-as-selector role described in
Section~\ref{sec:properties-oracle}.
Both \textbf{Mirror} and \textbf{Universalizability} are
conditionally cooperative strategies that recognize each other as
such and cooperate; this mutual recognition through strategy
inspection is the core mechanism of the OSDG.


\section{Strategy Documents}
\label{sec:appendix-strategies}

This appendix reproduces, verbatim, the nine strategy documents of
Tournament~9 --- the complete input the oracle receives about each
player, apart from the game rules and the instruction to decide.
They are read directly from the released repository, so the
listings below cannot drift from the text that was actually
adjudicated.
Documents are grouped as in Table~\ref{tab:archetypes}, in order of
increasing sophistication of what the decision rule conditions on;
Markdown formatting is the authors' own.
Their brevity is itself a finding: the behavior analyzed throughout
this paper is generated by roughly five hundred words in total.

\paragraph{Unconditional.}
The recipient's strategy is ignored.

\strategylisting{Generous}{generous}
\strategylisting{Selfish}{selfish}

\paragraph{Trait tests.}
The recipient's document is classified, without reference to how it
would treat this strategy in particular.

\strategylisting{Intelligence}{intelligence}
\strategylisting{Anti-exploiter}{anti-exploiter}
\strategylisting{Universalizability}{universalizability}

\paragraph{Reciprocity.}
The condition is the recipient's decision toward oneself, which
makes the rule self-referential (Section~\ref{sec:properties-oracle}).

\strategylisting{Conditional cooperator}{conditional-cooperator}
\strategylisting{Mirror}{mirror}

\paragraph{Second-order norms.}
The condition is how the recipient treats third parties --- the
unconditional cooperator in particular --- which is what makes these
two documents mutually antagonistic
(Section~\ref{sec:second-order}).

\strategylisting{Cooperation coalition}{cooperation-coalition}
\strategylisting{Chivalry}{chivalry}

\end{document}